\documentclass[runningheads]{llncs}

\usepackage{hyperref}
\usepackage[T1]{fontenc}
\usepackage{graphicx}
\usepackage{amsfonts}
\usepackage{xspace}
\usepackage{amsmath,amssymb}
\usepackage{xcolor}
\usepackage{esvect}
\usepackage{complexity}
\newclass{\PTIME}{PTIME}
\definecolor{purple}{rgb}{0.65, 0.27, 0.95}
\definecolor{rose}{rgb}{0.8, 0, 0.4}
\definecolor{vert}{rgb}{0, 0.6, 0.4}

\definecolor{ashgrey}{rgb}{0.7, 0.75, 0.71}
\definecolor{iceberg}{rgb}{0.44, 0.65, 0.82}
\definecolor{pastelorange}{rgb}{1.0, 0.7, 0.28}
\definecolor{lightkhaki}{rgb}{0.94, 0.9, 0.55}

 \usepackage{array}
\newcolumntype{C}[1]{>{\centering\let\newline\\\arraybackslash\hspace{0pt}}m{#1}}
\usepackage{booktabs}
\newsavebox\MBox

\usepackage{colortbl} 
\usepackage{multirow}

\newcommand{\B}{\mathbb{B}}
\newcommand{\set}[1]{\ensuremath{\left\{#1\right\}}}
\newcommand{\ie}{\emph{i.e.}\@\xspace}
\newcommand{\inmts}{\text{IN\_MTS}}
\newcommand{\under}{\phi_u}
\newcommand{\modelu}{\mu}
\newcommand{\raff}{\phi_{r}}
\newcommand{\contr}{\phi_{ce}}

\usepackage{algorithm}
\usepackage{algorithmic}

\newcommand*{\descbox}[1]{%
   \mathop{\boxed{#1}}\limits
}
\usepackage{dashbox}

\newcolumntype{U}{>{\centering\arraybackslash}p{3em}}
\newcolumntype{T}{>{\centering\arraybackslash}p{2em}}

\renewcommand\orcidID[1]{}

\begin{document}
\title{%
Tackling Universal Properties of Minimal Trap Spaces of Boolean Networks
}

%
%
\author{%
Sara Riva\inst{1}\orcidID{0000-0003-2133-8089}\and
Jean-Marie Lagniez\inst{2}\orcidID{0000-0002-6557-4115}\and
Gustavo Magaña López\inst{1}\orcidID{0000-0001-8140-0468}\and
Lo{\"\i}c Paulev{\'e}\inst{1}\orcidID{0000-0002-7219-2027}
}
\authorrunning{S. Riva et al.}
%
\institute{%
Univ. Bordeaux, CNRS, Bordeaux INP, LaBRI, UMR 5800, F-33400 Talence, France \\
\email{\{sara.riva,gustavo.magana,loic.pauleve\}@labri.fr}
\and
Univ. Artois, CNRS, CRIL, F-62300 Lens, France \\
\email{lagniez@cril.fr}
}

\maketitle              
\begin{abstract}
Minimal trap spaces (MTSs) capture subspaces in which the Boolean dynamics is trapped, whatever the update
mode. They correspond to the attractors of the most permissive mode.
Due to their versatility, the computation of MTSs has recently gained traction, essentially by
focusing on their enumeration.
In this paper, we address the logical reasoning on universal properties of MTSs in the scope of two
problems: the reprogramming of Boolean networks for identifying the permanent freeze of Boolean variables
that enforce a given property on all the MTSs,
and the synthesis of Boolean networks from universal properties on their MTSs.
Both problems reduce to solving the satisfiability of quantified propositional logic formula
with 3 levels of quantifiers ($\exists\forall\exists$).
In this paper, we introduce a Counter-Example Guided Refinement Abstraction (CEGAR) to efficiently
solve these problems by coupling the resolution of two simpler formulas.
We provide a prototype relying on Answer-Set Programming for each formula and show its tractability
on a wide range of Boolean models of biological networks.

\end{abstract}

\section{Introduction}

Since recent years, we observe a surge of successful applications of Boolean networks (BNs) in biology and
medicine for the modeling and prediction of cellular dynamics in the case of cancer and cellular
reprogramming~\cite{zanudo21,Reda2022,montagud22}.
Such applications face two main challenges: being able to design a qualitative Boolean model which
is faithful to the behavior of the biological system and being able to compute predictions to control its
(long-term) dynamics.
From a computational point of view, the latter problem mostly depends on the complexity of the dynamical
property to enforce, while the former additionally suffers from the combinatorics of candidate
models.

In a BN, the state of interacting components are modeled with Boolean values $\B =\{0,1\}$.
Then, given $n$ components, the dynamics evolve within the finite discrete space of
\emph{configurations} $\B^n$.
For each component $i$, the BN specifies a Boolean function $f_i:\B^n\to \B$ to compute the value
towards which the component state evolves from a given configuration of the network.
The transitions between the configurations are then computed according to an \emph{update mode}.
For instance, with the synchronous mode, a configuration $x\in\B^n$ has a (unique) transition
to configuration $(f_1(x),\cdots,f_n(x))$,
whereas, with the fully asynchronous mode, it has one transition for each component $i$ such that
$f_i(x)\neq x_i$ and going to $(x_1,\cdots,x_{i-1},f_i(x),x_{i+1},\cdots,x_n)$.
There is a vast zoo of update modes defined in the literature. They reflect different modeling
hypotheses on how the components evolve with respect to each other, and can have a great impact on
the resulting dynamics~\cite{pauleve2022boolean}.
These update modes can be compared using a simulation relation:
an update mode simulates another if, for any transition $x\to y$ of the latter, there
exists a trajectory from $x$ to $y$ with the former mode.
This results in a hierarchy of update modes~\cite{pauleve2022boolean},
where the \emph{most permissive}~\cite{naturepauleve,automata21} simulates all Boolean update modes.
The most permissive mode captures any trajectory of any quantitative model
which is a refinement of the BN (intuitively, a refinement adds quantitative information on
interaction thresholds and state, while respecting the logic of state change).
Hence, most permissive Boolean dynamics have formal connections with quantitative systems, contrary to
(a)synchronous modes, which are known to preclude the prediction of actually feasible trajectories
in biological systems~\cite{naturepauleve}.
 
Most applications of BNs to biological systems involve two types of dynamical properties:
the trajectories between configurations, which model changes of the state of components over time,
and the attractors, which capture the long-term dynamics of the system.
An attractor can be characterized by a set of configurations from which there is no out-going
transition, and such that there is a trajectory between any distinct pair of its configurations.
When it is composed of a single configuration, the attractor is called a fixed point of the
dynamics.

Capturing properties that are shared by \emph{all} the attractors (or \emph{all} attractors
reachable from a given set of configurations) is, therefore, a fundamental task of BN modeling.
In this paper, we focus on two problems related to these universal properties:
the reprogramming of a given BN with the permanent freeze of components of the network,
and the synthesis of a BN which matches with a given architecture while showing the desired universal
property on its attractors.
These are relevant problems since attractors usually capture biological phenotypes \cite{cifuentes2022control,klarner2018basins} and the reprogramming of a biological phenomena is linked to the search of treatments \cite{montagud2022patient}. However, the computational complexity of these problems is stirred by the complexity of characterizing (all)
the attractors of a BN.
This complexity depends on the update mode.
For (a)synchronous update modes, determining whether a configuration belongs to an attractor is an
infamous PSPACE-complete problem, which largely impedes the tractability of analysis of networks
with several hundreds of components.
Indeed, attractors can have very different shapes with these modes.

The \emph{minimal trap spaces} of BNs are related to their attractors but do not depend on the update mode.
A \emph{trap space} is a subcube of $\B^n$ which is closed by the local functions
(the image by $f_1, \ldots, f_n$ of a vertex is one of its vertices).
It is minimal whenever there is no other trap space within it.
The fixed points are particular cases of minimal trap spaces.
In some sense, a trap space delimits a portion of the space from where any trajectory with any update
mode is trapped within.
Thus, a minimal trap space encloses at least one attractor, with any update mode.
However, an (a)synchronous (non-fixed point) attractor is not necessarily included in a minimal trap space.
Nevertheless, in practice for biological models, minimal trap spaces have been observed to be good approximations of
asynchronous attractors~\cite{klarner2015approximating}.
Moreover, it turns out that minimal trap spaces are exactly the attractors of the most permissive
mode~\cite{naturepauleve}.
Back to our computational point of view, (minimal) trap spaces are more amenable
objects thanks to a much lower complexity~\cite{moon2022computational}.
Recent approaches demonstrated the tractability of methods based on solving the satisfiability of
logical formulas for enumerating minimal trap spaces in BNs with several thousands of
components~\cite{naturepauleve,trinh2022minimal}.
In large networks, however, their exhaustive enumeration can be intractable.

In this paper, we study the problem of reprogramming and synthesis of BNs from properties on
    their minimal trap spaces. In the case of the most permissive update mode, this gives exact
    methods to reason on the attractors. In the case of other update modes, this gives approximate
methods to address their attractors.
Specifically, we will consider \emph{marker} properties of trap spaces:
a marker is a partial map associating a subset of components with a Boolean value, e.g. $\{ a\mapsto
1, c\mapsto 0\}$ where $a$, and $c$ are components of the BN.
A trap space matches with a marker if all its configurations match with it (e.g., $a$ is always $1$ and $c$
always $0$ in the trap space).
The \emph{marker reprogramming}~\cite{pauleve2022marker} consists in permanently freezing a subset
of its components to specific Boolean values so that all the minimal trap spaces of the resulting BN match with the
marker.
The \emph{synthesis} consists in deriving a BN that matches with a given network architecture
(influence graph) and such that all its minimal trap spaces match with a given marker.
These problems can be expressed as a logical formula of the form
``there is a permanent freeze $P$ (resp. a BN matching with influence graph) such that all the
minimal trap spaces of the BN perturbed by $P$ (resp. of the BN) match with the given marker $M$''.
As we will explain, both problems boil down to solving the satisfiability of
quantified propositional Boolean formulas (QBF)~\cite{Bning2021TheoryOQ} with three levels of quantifiers ($\exists\forall\exists$, 3-QBF).

While modern SAT~\cite{booksat} and Answer-Set Programming (ASP)~\cite{baral2003knowledge,gebser2012answer} solvers can address efficiently
1-QBF (NP) and 2-QBF problems respectively, the generic solving of higher order QBFs problems turns out
to be very challenging.
In \cite{pauleve2022marker}, the marker reprogramming of minimal trap spaces is tackled by solving a complementary
problem which is only 2-QBF.
However, as we will show in experiments, this approach turns out to be intractable for large
networks and for the synthesis problem.
To the best of our knowledge, this is so far the only other method addressing universal properties
over minimal trap spaces in BNs.

Instead of solving directly the 3-QBF problems, we introduce in Sect.~\ref{sec:cegar} a logic
approach based on a Counter-Example-Guided Abstraction Refinement (CEGAR) of a simpler formula.
Essentially, we extract candidate perturbations (resp. BNs) from an NP formula and verify
using a 2-QBF formula whether they fulfill the universal property.
If not, we extract a counter-example that we generalize and plug in the
original NP formula. The procedure is repeated until either we prove that the 3-QBF problem is
not satisfiable, or a candidate perturbation (resp. BN) verifies the universal property on the minimal trap spaces.
We developed a prototype based on ASP and show in Sect.~\ref{sec:experiments} its tractability for
the reprogramming and synthesis of large BNs.

\section{Preliminaries}
\subsection{Boolean Networks and Trap Spaces}

A BN of dimension $n$ is a function $f:\B^n \rightarrow \B^n$ where $\B=\{0,1\}$. 
The vectors $x\in \B^n$ are its \emph{configurations}, where, for each $i\in\{1,\ldots, n\}$, $x_i$
denotes the \emph{state} of \emph{component} $i$.
For each component $i$, $f_i:\B^n \rightarrow \B$ is called its \emph{local function}.
It can be specified using  truth tables, Binary Decision Diagrams (BDDs)
\cite{drechsler2013binary}, or propositional formulas, to name but a few.
Each $f_i$ typically depends on a subset of components of the BN.
The \emph{influence graph} $G(f)$ captures these dependencies.
It is the signed digraph $(\{1,\ldots,n\},E_+\cup E_-)$ such that
there is a positive (resp. negative) influence of $i$ on $j$, \ie, $(i,j)\in E_+$ (resp. $E_-$) if and only if there
exists at least one configuration $x$ such that
$f_j(x_1,\ldots,x_{i-1},1,x_{i+1},\ldots, x_n) - f_j(x_1,\ldots,x_{i-1},0,x_{i+1},\ldots, x_n) = 1$
(resp. $-1$).
Remark that different BNs can have the same influence graph.
A BN $f$ is \emph{locally monotone} whenever $E_+\cap E_-=\emptyset$: a component $i$ cannot influence
positively and negatively a component $j$.

A vector $h\in \set{0,1,\ast}^n$ denotes a subcube of $\B^n$ where dimensions with value $\ast$ are
\emph{free}, and others are \emph{fixed}.
Its vertices are the $2^k$ configurations $c(h)=\set{x\in\B^n \mid h_i \neq \ast \Rightarrow
x_i=h_i}$ where $k$ is the number of free dimensions.
A subcube $h$ is a \emph{trap space} if it is closed by $f$, \ie, $\forall x\in c(h)$, $f(x)\in c(h)$.
Note that $(\ast)^n$ is always a trap space.
Given two subcubes $h,h'$, $h$ is \emph{smaller than} $h'$, noted $h \preceq h$', if and only if
$c(h) \subseteq c(h')$.
A trap space is \emph{minimal} if it does not contain a smaller trap space.
We denote by $TS_f(x)\in \set{0,1,\ast}^n$ the smallest trap space of $f$ containing the
configuration $x$.
$TS_f(x)$ always exists and is unique: if two subcubes $h,h'$ are trap spaces, their intersection is a trap space.

\begin{example}\label{ex:back}
Consider the BN $f:\B^4\to\B^4$ with
$f_1(x)=x_2$, $f_2(x)=x_1$, $f_3(x)=\neg x_4 \land (x_1 \lor x_2)$, and $f_4(x)=\neg x_3$.\\
\centerline{\includegraphics[width=0.35\textwidth]{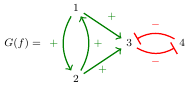}}
It is locally monotone and $h=11\ast\ast$ is a trap space since
$$\{f(1100),f(1101),f(1110),f(1111)\}\subseteq c(11\ast\ast) = \{1100,1101,1110,1111\}$$
but $h$ is not minimal since it contains the (minimal) trap space $1101$.
\end{example}

\subsection{The Marker Reprogramming Problem}\label{sec:rep}

We assume the dimension $n$ of BNs fixed.
We denote by $\mathbb M$ the set of all partial maps from $\{1,\ldots,n\}$ to $\B$.
Given $k\in\mathbb N$, we write $\mathbb M^{\leq k}$ the partial maps with at most $k$ associations.
In the following, \emph{markers} and \emph{perturbations} are defined as partial maps.
Notice that a partial map is equivalent to a subcube where the fixed dimensions are specified
by the mapping.
We say that a configuration $x\in\B^n$ \emph{matches} with a marker $M$, denoted by $x\models M$,
whenever for each $i\mapsto b \in M$, $x_i=b$.
Similarly, given a subcube $h\in\{0,1,\ast\}^n$, $h\models M$ if and only if for each $i\mapsto b\in
M$, $h_i=b$. Equivalently, $h\models M \Leftrightarrow \forall x\in c(h), x\models M$.

Given a BN $f$ and a perturbation $P\in \mathbb M$, the \emph{perturbed BN} $f/P$ is obtained by
replacing the corresponding local functions with constant values: for each component
$i\in\set{1,\ldots,n}$,
$$
(f/P)_i(x)=\begin{cases}b & \text{if $i\mapsto b\in P$,}\\ 
f_i(x) & \text{otherwise.}
\end{cases}
$$
Intuitively, a perturbation permanently freezes involved components.
It is important to remark that the minimal trap spaces of $f/P$ and $f$ can be very different.

\smallskip
Given a BN $f$ and a marker $M\in\mathbb M$, the \emph{marker reprogramming} problem consists in
identifying perturbations $P\in\mathbb M$ such that \emph{all} the minimal trap space of $f/P$
match with $M$.
Typically, we aim the complete identification of (subset)minimal perturbations only, \ie, the perturbations
so that no submap is a solution.
Moreover, it is usual to limit the size of the perturbations to some constant $k$ as many
simultaneous perturbations can be difficult to implement experimentally.
Similarly, some components may be \emph{uncontrollable}, and the perturbations must not involve
them.
With either of these cases, the problem can be non-satisfiable (otherwise $P=M$ is a trivial
solution), and deciding its satisfiability is already challenging.
Finally, notice that if $P=\emptyset$ is a solution, then all the minimal trap spaces of $f$ match with $M$.

Marker reprogramming generalizes the fixed point reprogramming considered in
\cite{biane2018causal,moon2022bilevel}, limited to ensuring that all the fixed points only
match with the marker.
In \cite{pauleve2022marker}, the problem over minimal trap spaces has been tackled for the first
time in the scope of locally monotone BNs.
The proposed approach relies on a modeling of the problem in QBF.
Let us first consider the predicate $\inmts_{f/P}(x)$ which is true if and only if $x$ belongs to a
minimal trap space of the BN $f/P$. The marker reprogramming problem can then be expressed as
follows:
$$
\exists P \in \mathbb{M}^{\leq k}, \forall x \in\B^n,  \inmts_{f/P}(x) \Rightarrow x \models M
$$
In the equation, $\inmts_{f/P}(x)\Rightarrow x \models M$ can be reformulated as $x\models M \vee
\neg\inmts_{f/P}(x)$. Then, one can remark that $\inmts_{f/P}(x)$ is false if and only if $TS_{f/P}(x)$ is not
minimal, \ie, there exists a configuration $y\in TS_{f/P}(x)$ such that $TS_{f/P}(y)\subsetneq TS_{f/P}(x)$.
In the case of locally monotone BNs, $TS_{f/P}(x)$ and $TS_{f/P}(y)$ can be computed in polynomial
time~\cite{naturepauleve}.
Thus, the marker reprogramming boils down to the following 3-QBF:
\begin{equation}\label{p3}
    \exists P \in \mathbb{M}^{\leq k}, \forall x \in\B^n : x \not\models M,  \exists y \in
    TS_{f/P}(x), TS_{f/P}(y) \neq TS_{f/P}(x).
\end{equation}

The approach of \cite{pauleve2022marker} relies on Answer-Set Programming (which is limited to 2-QBF
problems) to solve the complementary problem of identifying all perturbations $P$ such that at least one minimal trap
space of $f/P$ does not match with the marker
($\exists P \in \mathbb{M}^{\leq k}, \forall x \in\B^n, x\not\models M\wedge \forall y\in TS_{f/P}(x), TS_{f/P}(y)=TS_{f/P}(x)$).
Then, the solutions are obtained by an ensemble difference with $\mathbb M^{\leq k}$.
While $\mathbb M^{\leq k}$ is of polynomial size with $n$, this complementary problem becomes rapidly
intractable with large networks having numerous wrong perturbations.

\begin{example}\label{ex:MTS}
Consider the BN $f$ with $f_1(x)=\neg x_2$, $f_2(x)=\neg x_1$, $f_3(x)=x_1 \land \neg x_2 \land \neg x_4$, $f_4(x)= x_3 \lor x_5$ and $f_5(x)=\neg x_3\land x_5$.
It has two minimal trap spaces ($010\ast\ast$ and $10\ast\ast\ast$).
If we consider the marker $M=\set{2\to1,3\to1}$ and $k=2$, a possible solution is
$P=\set{3\to1,1\to0}$:
the BN $f/P$ has just a single minimal trap space, $01110$, which matches with $M$.
\end{example}

\subsection{The Synthesis Problem}

The automatic design of BNs from specifications on their static and dynamical properties is another
prime challenge for applications in biology
\cite{dorier2016boolean,yordanov2016method,Chevalier2020-CMSB}.

There has been recent progress to address the synthesis from asynchronous dynamical
properties, including attractors~\cite{Goldfeder2019,Benes2020}, but they still show limited tractability.
Synthesis of BNs from specifications on their most permissive dynamics has been shown to have 
great scalability, thanks to a lower computational complexity, with applications to networks up to several
thousands of components~\cite{chevalier2019synthesis}.
Nevertheless, prior work did not account for universal properties over minimal trap spaces.

Static properties of the network allow delimiting the domain of candidate BNs, which we denote by $\mathbb F$.
It usually comes from a given signed influence graph $\mathcal G=(\set{1,\ldots,n},\mathcal
E_+\cup\mathcal E_-)$. In that case, $\mathbb F$ could be, for instance, any BN $f$ such that its influence graph
$G(f)=(\set{1,\ldots,n},E_+\cup E_-)$ is equal to $\mathcal G$, or included in $\mathcal G$
(\ie, $E_+\subseteq \mathcal E_+$ and $E_-\subseteq \mathcal E_-$).
Additionally, $\mathbb F$ could be restricted with already specified partial Boolean local
functions.

\smallskip
Given a domain of BNs $\mathbb F$ and a marker $M\in\mathbb M$, the \emph{synthesis problem} consists
in identifying a BN $f\in \mathbb F$ such that \emph{all} the minimal trap spaces of $f$ match with
$M$.
It can be expressed as 3-QBF in a very similar fashion to the marker
reprogramming problem~\eqref{p3}:
\begin{equation}\label{synt}
\exists f\in\mathbb F, \forall x \in\B^n : x \not\models M,  \exists y \in TS_f(x), TS_f(y) \neq TS_f(x).
\end{equation}
Note this problem can be unsatisfiable.
The main difference with marker reprogramming is the combinatorics of the domain $\mathbb F$.
To our knowledge, there was no approach to tackle this synthesis problem efficiently.

\begin{example}\label{ex:syn}
Consider $M=\set{3 \to 1}$ and $\mathcal G$ equal to the $G(f)$ presented in Example \ref{ex:back}. 
A solution to the synthesis problem is $f_1(x)=1$, $f_2(x)=x_1$, $f_3(x)= x_2 \lor \neg x_4$ and $f_4(x)=\neg x_3$.
In this scenario, the influence graph is a sub-graph of $\mathcal G$ and the only minimal trap space (which also respects $M$) is $1110$.
\end{example}

\subsection{Counter-Example-Guided Abstraction Refinement (CEGAR)}\label{sec:intro-cegar}

CEGAR~\cite{clarke2003counterexample} is an incremental way to decide the satisfiability of a (possibly quantified) logic formula
$\phi$ by the mean of a simpler formula $\under$ (resp. $\phi_o$) so that the models of $\under$
subsume (resp. $\phi_o$ are subsumed by) the models of $\phi$.
Thus, $\under \Leftarrow \phi$ (resp. $\phi_o \Rightarrow \phi$).
The choice between $\under$ and $\phi_o$ depends on whether one wants to \emph{under-} or \emph{over-approximate} $\phi$.
%

We briefly explain the principle with the $\under$ case.
If $\under$ is unsatisfiable, so is $\phi$.
Otherwise,
a model $\vv{\modelu}$ of $\under$ is found ($\vv{\modelu}\models \under$), and one must verify whether $\vv{\modelu}\models \phi$.
If $\vv{\modelu}\not\models\phi$, $\vv{\modelu}$ is a \emph{counter-example} and $\under$ must be refined
with some $\raff(\vv{\modelu})$ so that $\under\land\raff(\vv{\modelu})\Leftarrow \phi$.
The process is repeated until the refined $\under$ is either unsatisfiable,
demonstrating the unsatisfiability of $\phi$,
or a model of the refined $\under$ is a model of $\phi$.
The challenge is thus to design a refinement which makes this process converge rapidly, which is
problem-specific.


\section{A CEGAR for Minimal Trap Spaces}\label{sec:cegar}

In this section, we introduce a CEGAR-based approach for addressing universal marker properties over
minimal trap spaces.
The refinement can be directly employed in solving the marker reprogramming~\eqref{p3} and
synthesis~\eqref{synt} problems.
In the case of locally monotone BNs, or BNs whose local functions are specified using propagation
complete representations (such as BDDs or Petri net), this boils down to iteratively solving on the
one hand an NP problem ($\under$) and on the other
hand a 2-QBF problem for identifying counter-examples.

We first detail the CEGAR for the case of marker reprogramming before discussing its generalization
to the synthesis problem.

\subsection{Generalizing counter-examples for refinement}\label{sec:ref}

Recall that the marker reprogramming problem consists in identifying perturbations of size at most
$k$ under which all the minimal trap spaces of the given BN $f$ match with the given marker $M$.
Let us consider the 3-QBF~\eqref{p3}, that we refer to as $\phi$ in this
section, following the notations of Sect.~\ref{sec:intro-cegar}.
Let us assume that we are given a candidate perturbation $\vv P\in\mathbb M^{\leq k}$, for instance, from a
model of the NP formula $\under=\exists P\in\mathbb M^{\leq k}$.
To be a solution for $\phi$, one must verify that all the minimal
trap spaces of the perturbed BN $f/\vv P$ match with the marker.
Thus, $\vv{P}$ and $\vv{x}$ are used to represent a specific perturbation and a specific configuration.
A \emph{counter-example} would be a configuration  of a minimal trap space of $f/\vv P$
that does not match with the marker.
Such counter-examples are models of the following QBF:
\begin{equation}\label{eq:counter}
    \contr(\vv P) = \exists x \in\B^n \text{ s.t. } x \not\models M \text{ and } \forall y \in TS_{f/\vv P}(x), TS_{f/\vv P}(x)=TS_{f/\vv P}(y).
\end{equation}
If $\contr(\vv P)$ is not satisfiable, then $\vv P$ is a model of $\phi$ and thus a solution to the
marker reprogramming problem.
Otherwise, let us denote by $\vv{x}$ a model of $\contr(\vv P)$, \ie, a counter-example showing that $\vv P$
is not a model of $\phi$:
the configuration $\vv{x}$ belongs to a minimal trap space of $f/\vv P$ and does not match with the marker.
Intuitively, $\vv{x}$ is a configuration that shows why $\vv P$ is not valid with respect to a marker $M$.
The idea to move forward in the search for a valid perturbation is to avoid any other perturbation
$P\in\mathbb{M}^{\leq k}$ such that $\vv x$ belongs to one of its minimal trap spaces.
Let us now point out some useful properties concerning trap spaces and markers:
\begin{property}\label{prop:general}
For any BN $f:\B^n \to \B^n$, marker $M\in\mathbb{M}^{\leq n}$, perturbation $P\in\mathbb{M}^{\leq k}$
and configuration $x\in\B^n$:
\begin{enumerate}
    \item if $TS_{f/P}(x)\models M$, all minimal trap spaces within $TS_{f/P}(x)$ match with $M$;\label{prop:mts}
    \item if $x \not\models M$, it holds that $TS_{f/P}(x) \not\models M$;\label{prop:marker}
\item if $x \not\models M$, any perturbation $P'$ such that $x$ is in a minimal trap space of
    $f/P'$ is not a model of $\phi$;\label{prop:counter}
\item if $x$ is not in a minimal trap space of $f/P$, there exists a configuration $y\in TS_{f/P}(x)$ such that $TS_{f/P}(y)\varsubsetneq TS_{f/P}(x)$.\label{prop:min}
\end{enumerate}
\end{property}
As illustrated by Fig.~\ref{fig:counterexample}, these properties imply two constraints that must be
verified by any perturbation $P$ model of $\phi$:
\begin{figure}[tb]
    \centering
    \includegraphics[height=30mm]{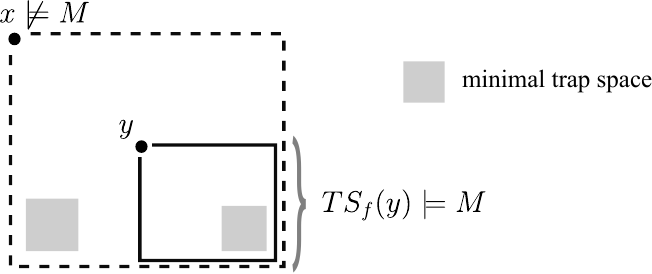}
    \caption{\label{fig:counterexample}
        Illustration of the refinement given an $x\in\B^n$ such that $x\not\models M$.
        If all the minimal trap spaces (gray squares) of $f$ match with $M$, there must
        exist $y\in TS_f(x)$ (dashed line) so that all the configurations of
        $TS_f(y)$ (plain line) match with $M$.}
\end{figure}
(a) the trap space $TS_{f/P}(\vv{x})$ is not minimal, \ie, there must exist a configuration $y\in TS_{f/P}(\vv x)$
whose trap space is strictly smaller ($TS_{f/P}(y)\subsetneq TS_{f/P}(\vv x)$);
and (b) $TS_{f/P}(\vv x)$ must contain a trap space which matches with the marker $M$.
Combining (a) and (b), there must exists at least one configuration $y\in TS_{f/P}({\vv x})$ such that $TS_{f/P}(y)\models M$.
Therefore, we define the following refinement from the counter-example $\vv x$:
\begin{equation}\label{eq:p3-raf}
\raff(\vv{x}) = \exists y \in \B^n: TS_{f/P}(y) \varsubsetneq TS_{f/P}(\vv{x}) \land TS_{f/P}(y)\models M\enspace.
\end{equation}

Remarking that $\vv P\not\models \under\land \raff(\vv{x})$, one can then apply the CEGAR approach
by iterating the refinement until either no counter-example can be found (and thus $\vv P$ is a
solution), or until the refined formula becomes non-satisfiable.
The correctness is expressed by the following lemma:

\begin{lemma}\label{lem:cegar}
Given $\vv{P} \models \under$ and $\vv{x}\models \contr(\vv{P})$, it holds that $\under \land \raff(\vv{x}) \Leftarrow \phi$.
\end{lemma}
\begin{proof}
Let us show that, starting from an under-approximation $\under$ and adding $\raff$, we get a new under-approximation of $\phi$.
In other words, we want to show that the set of solutions is shrunk by removing only invalid solutions.
Initially, $\under = \exists P\in \mathbb{M}^{\leq k}$. Then, $\under \Leftarrow \phi$. 
Considering a model $\vv{P}$ of $\under$ that is not a model of $\phi$,
one can find a model $\vv{x}$ of $\contr(\vv{P})$.
The identification of $\vv{x}$ is interesting since all $P$ in $\mathbb{M}^{\leq k}$ such that $\inmts_{f/P}(\vv{x})$ are not valid solutions (point \ref{prop:counter} of Property \ref{prop:general}).
By initially defining $\raff(\vv{x}) = \exists y \in \B^n: TS_{f/P}(y) \varsubsetneq
TS_{f/P}(\vv{x})$ we remove all candidates $P$ for which $\vv{x}$ would be in a minimal trap space (point~\ref{prop:min} of Property \ref{prop:general}). 
The solution space is reduced by removing only invalid perturbations.
Adding $TS_{f/P}(y) \models M$, we impose that at least one minimal trap space within $TS_{f/P}(\vv{x})$
matches the marker (point~\ref{prop:mts} of Property \ref{prop:general}).
This constraint must indeed be satisfied in the solutions of $\phi$.
It allows to eliminate all $P$ for which $\nexists y \in \B^n: TS_{f/P}(y) \varsubsetneq TS_{f/P}(\vv{x}) \land TS_{f/P}(y) \models M$.
In other words, all $P$ where we fail to ensure that all minimal trap spaces in $TS_{f/P}(y)$ match the marker. Recall that the configuration $y$ is needed to ensure $\vv{x}$ outside minimal trap spaces. 
Again only invalid solutions are removed and thus a new under-approximation is obtained.
In a generic step, we have $\under = \exists P\in \mathbb{M}^{\leq k} \land \raff(\vv{x}^{(1)})
\land \ldots \land \raff(\vv{x}^{(q)})$. The solution space is further reduced with
$\raff(\vv{x}^{(q+1)})$, requiring that $\vv{x}^{(q+1)}$ is not within a minimal trap space and
that $TS_{f/P}(\vv x^{(q+1)})$ contains a trap space matching with the marker.
\end{proof}

The actual complexity of derived QBF formulas largely depends on the encoding of the $TS$ predicate,
and is discussed in Sect.~\ref{sec:complexity}.

\subsection{Necessary condition on perturbations}\label{sec:firstqbf}

In the previous section, the initial perturbation candidate is a model of $\phi_u=\exists
P\in\mathbb M^{\leq k}$.
This can be already refined by remarking that there always exists at least one minimal trap space in any BN.
Thus, one can already impose that at least one minimal trap space of $f/P$ matches with the marker.
By Prop.~\ref{prop:mts}, this is ensured by the existence of a configuration $w$ such that
$TS_{f/P}(w)\models M$.
This leads to the following formula:
\begin{equation}\label{eq:under}
\under = \exists P \in \mathbb{M}^{\leq k}, \exists w \in\B^n, TS_{f/P}(w)\models M\enspace.
\end{equation}

Remark that with this version of $\under$, only the first iteration of the CEGAR can be affected as
the existence of $w$ is subsumed by the refinement $\raff$.
Therefore, Lemma~\ref{lem:cegar} still holds with this $\under$.
Algorithm~\ref{alg:cegarR} summarizes the overall CEGAR procedure for solving the marker
reprogramming problem.

\floatname{algorithm}{Algorithm}
\renewcommand{\algorithmicrequire}{\textbf{Input:}}
\renewcommand{\algorithmicensure}{\textbf{Output:}}
\begin{algorithm}[tb]
\caption{CEGAR-based Reprogramming}
\begin{algorithmic}[1]\label{alg:cegarR}
\REQUIRE $f:\B^n \to \B^n, M \in \mathbb{M}^{\leq n}, k \in \set{0,\ldots,n}$
\ENSURE $P\in\mathbb{M}^{\leq k} \text{ such that } \forall x\in\B^n: \inmts_{f/P}(x) \Rightarrow x\models M $
\STATE $\under\;=\exists P \in \mathbb{M}^{\leq k}, \exists w \in\B^n : TS_{f/P}(w)\models M$
\STATE $\vv{P}\;=$ solve$(\under)$
\WHILE {$\vv{P}$ exists}
	\STATE {$\contr(\vv{P}) \;= \exists x \in\B^n : x \not\models M \land \inmts_{\vv{P}}(x)$}
	\STATE {$\vv{x}\;=$ solve($\contr(\vv{P})$)}
	\IF{$\vv{x}$ exists}
		\STATE $\raff(\vv{x})\;= \exists y\in\B^n : TS_{f/P}(y)\varsubsetneq TS_{f/P}(\vv{x}) \land TS_{f/P}(y) \models M$
		\STATE $\under\;= \under \land \raff(\vv{x})$
		\STATE $\vv{P}\;=$ solve$(\under)$
	\ELSE 
		\RETURN $\vv{P}$
	\ENDIF
\ENDWHILE
\RETURN UNSAT
\end{algorithmic}
\end{algorithm}
\begin{example}
Consider the BN of Example \ref{ex:back}.
According to $\under$ \eqref{eq:under}, a first candidate solution is $\vv{P}=\emptyset$.
To verify if all minimal trap spaces of $f/\vv{P}$ (\ie, $f$) match with $M=\set{1\to 1, 3\to 1}$, we verify the existence of a counter-example (with $\contr(\vv{P})$) and identify its model $\vv{x}=1101$.
In fact, $1101$ is a fixed point of $f/\vv{P}$.
Then, the new under-approximation is $\under \land \raff(1101)$ and we can identify a new candidate perturbation $\vv{P}=\set{4\to 0}$.
In this case, a counter-example is $\vv{x}=0000$.
Improving the approximation with $\under\land\raff(1101)\land\raff(0000)$, we can identify now $\vv{P}=\set{2\to 1,4\to 0}$ and no counter-example can be found.
Then, $P=\set{2\to1,4\to 0}$ is a solution of the reprogramming problem.
\end{example}

\subsection{Generalization to the synthesis problem}\label{sec:syn}

Recall that the synthesis problem consists in identifying a BN $f$ in a given domain $\mathbb F$ so
that all its minimal trap spaces match with the given marker $M$.
This problem can be seen as a generalization of the marker reprogramming problem by considering
$\mathbb F = \{ f/P \mid P\in\mathbb M^{\leq k}\}$.

The CEGAR developed in the previous section can be straightforwardly applied to the synthesis problem:
\begin{align}
    \under' &= \exists f \in \mathbb{F}, \exists w \in\B^n, TS_f(w)\models M \\
    \contr'(\vv f) &= \exists x \in\B^n \text{ s.t. } x \not\models M \text{ and } \forall y \in
                        TS_{\vv{f}}(x), TS_{\vv{f}}(x)=TS_{\vv{f}}(y)\\
    \raff'(\vv{x}) &= \exists y \in \B^n: TS_f(y) \varsubsetneq TS_f(\vv{x}) \land TS_f(y)\models M\enspace.
\end{align}
BN candidates are models of $\under'$ where, as for the reprogramming case, we already enforce the existence of a
minimal trap space matching with the marker.
If $\vv f$ is a model of $\under'$, $\contr'(\vv f)$ characterizes the counter-examples
configurations, and $\raff'(\vv x)$ provides the refinement from such a given counter-example $\vv
x$.

\begin{example}
Consider a complete graph $\mathcal{G}$ such that all edges are positive.
It is known that any BN $f$ having $G(f)=\mathcal G$ has only two fixed points ($0^n$ and $1^n$) and
no cyclic asynchronous attractor~\cite{aracena2008maximum,aracena2004positive}.
Given any $M$, the synthesis problem results in expression \eqref{synt}.
Let us start by considering the under-approximation 
$$\under' = \exists f \in\mathbb{F}, \exists w \in\B^n : TS_f(w) \models M.$$
Let us assume $M=\set{1\to0}$ and $n=3$.
Any $f$, such that $G(f)=\mathcal{G}$, is a valid candidate solution $\vv{f}$ since $x=000$ is a fixed point.
However, searching for a counter-example $\vv{x}$, we will find $\vv{x}=111$ since it is a fixed point and $\vv{x}\not\models M$.
With the refinement $\raff(\vv{x})$, the problem turns out to be 
\begin{equation}
\resizebox{.9\hsize}{!}{$\exists f \in\mathbb{F}, \exists w \in\B^n : TS_f(w) \models M \land (\exists y\in\B^n: TS_f(y)\varsubsetneq TS_f(111)\land TS_f(y) \models M)$}
\end{equation}
which is unsatisfiable.
Thus, the CEGAR approach allows us to verify the property without needing to solve the original problem \eqref{synt}.  
Remark that the known theoretical property is not implemented in the approach.
\end{example}

\subsection{Complexity}\label{sec:complexity}

It is important to notice that the formulas introduced in this section involve a predicate $TS_f(x)$ which
returns the smallest trap space containing the given configuration $x$ in the scope of BN $f$.
The encoding of this predicate can affect the complexity of the formulas, by adding variables,
but more importantly by potentially adding quantifiers.

As shown in~\cite{naturepauleve}, $TS_f(x)$ can be computed by progressively saturating a cube $h$,
starting from $h = x$:
for each component $i$ which is fixed in $h$, one check whether there exists a vertex of $h$, $y\in
c(h)$, such that $f_i(y)\neq h_i$. In that case, the $i$-th component is freed in the cube.
This iteration can be performed up to $n$ times.
Importantly, remark that the test $\exists y\in c(h): f_i(y)\neq h_i$ boils down to the SAT and UNSAT
problems of $f_i$.
However, it is known that the SAT/UNSAT decision can be deterministically computed in polynomial
time whenever $f_i$ is monotone (\ie, the BN $f$ is locally monotone), and, more generally,
whenever $f_i$ is given as a \emph{propagation complete} representation \cite{propagationref}, which includes BDDs
and Petri nets.
Therefore, in these cases, $TS_f(x)$ can be represented efficiently as a propositional formula (see
Appendix~\ref{sec:ts-cnf})
or using ASP, similarly to~\cite{chevalier2019synthesis}.
In other cases, the encoding $TS_f(x)$ involves the SAT problem.

We can then conclude that $\raff$ and the under-approximation $\under$ are NP ($\exists$) expressions,
while the check for counter-example $\contr$ is 3-QBF in the general case and 2-QBF with locally
monotone BNs or with propagation complete local functions (which are widely used by BN analysis
tools).

\section{Implementation and Performance Evaluation} \label{sec:experiments}

We implemented the CEGAR resolution of marker reprogramming and synthesis problems in a
prototype\footnote{
Code and dataset available at \url{https://github.com/bnediction/cegar-bonesis}}
relying on the Python library \textsc{BoNesis}\footnote{\url{https://github.com/bnediction/bonesis}} and
Answer-Set Programming multi-shot solver
\textsc{clingo}~\cite{Gebser2018}.
The prototype exploits a DNF representation for locally monotone local functions and a BDD
representation for non-monotone local functions (see Section \ref{sec:complexity}).

\subsection{Datasets}
We considered 2 sets of BNs and markers.

The \emph{Moon dataset}~\cite{moon2022bilevel}
consists of 10 locally monotone BNs and 1 non-monotone BN taken from biological
modeling literature.
Network sizes range from 13 to 75, classified as small (S1-S4), medium (M1-M3), and large (L1-L4).
Importantly, each BN comes with a marker and uncontrollable components from the related biological
application.

The \emph{Trappist dataset}~\cite{trinh2022minimal}
consists of 27 locally monotone BNs and 6 non-mono\-tones BNs
ranging from 47 to 4691 components taken from biological modeling literature, either designed
directly as BNs, or resulting from a conversion.
Contrary to the Moon dataset, no biologically-relevant markers are specified.
Therefore, for each network, we randomly generated two sets of markers
from vertexes of the bottom strongly connected components of the influence graphs (the ``output''
components), and associating a Boolean value to them.
For networks with at most 5 output components, all combinations of markers were considered.
Otherwise, we generated 100 markers, where each associates 3 random output components to a random Boolean value.
Then, for each network, a second set of markers of size 1 were generated with the same process.
Duplicates markers were removed.

\subsection{Protocol}

Besides assessing the scalability of our CEGAR implementation, we aimed at benchmarking it with
different variants of $\raff$ and $\under$, and in the case of reprogramming, with
already existing approaches.
Moreover, we aimed at evaluating generic QBF solvers, such as \textsc{CAQE}~\cite{rabe2015caqe} or
\textsc{DepQBF}~\cite{lonsing2017depqbf}, for tackling our 3-QBF problems, for which we devised a
standard quantified conjunctive normal form encoding (Appendix~\ref{app:qbf}).
However, it turned out that they failed to scale for the vast majority of our instances.

\paragraph{Marker reprogramming problem.}
The inputs were the BNs and their associated markers, together with the maximum size of
perturbations (parameter $k$).
Only subset-minimal perturbations were considered, using the adequate solving mode of
\textsc{clingo}.
Moreover, we denied perturbations involving components of the markers, and those
declared as uncontrollable in the Moon dataset.
In the case the instance is satisfiable, we analyzed both the time for computing the first solution, and
the time to exhaust the full set of subset-minimal solutions.
For the reprogramming scenario, we compared several methods:
\begin{itemize}
    \item Enumeration and filtering:
        the enumeration is performed by increasing size in order to obtain only subset-minimal
        solutions. The filtering is performed using $\contr$.
        This approach somehow corresponds to the most basic CEGAR implementation without any
        counter-example generalization 
        ($\raff$ is $P \not = \vv P$).
    \item Complementary: the method of~\cite{pauleve2022marker} based on enumerating
        perturbations that fail the reprogramming and subtract them from $\mathbb M^{\leq k}$
        (Sect.~\ref{sec:rep}).
    \item CEGAR-2: our implementation of Algorithm~\ref{alg:cegarR}.
    \item CEGAR-1:
        is the approach presented in Algorithm~\ref{alg:cegarR} with the difference that a weaker refinement, imposing only that 
        counter-examples are not part of minimal trap spaces, is used, \ie,
        $\raff(\vv{x})\;= \exists y\in\B^n : TS_{f/P}(y)\varsubsetneq TS_{f/P}(\vv{x})\enspace$ 
        (here it is not imposed that the trap space matches the marker).
\end{itemize}

The aim was therefore to investigate the scalability of the approach and the importance of the
various constraints added in CEGAR-1 and CEGAR-2, and compare them to the current state of the art.

\paragraph{Synthesis problem.}
The inputs were the influence graph of the BNs and their associated markers.
The domain $\mathbb F$ was defined as the set of locally monotone BNs whose local functions are
represented in disjunctive normal form (DNF).
With these settings, the problem can be unsatisfiable as local functions of components cannot be assigned to
constant functions, unless already a constant function in the original BN. Thus, the trivial
solution where marker components are assigned to their corresponding value is not part of the search domain
$\mathbb F$.
For the Moon dataset, we imposed that their influence graph match exactly with the input one; while
for the Trappist dataset, we relaxed this constrained by imposing only that the DNF of each local
function involves all the regulators of the corresponding component with the adequate sign, and
we limited the size of DNFs to 32 clauses.
Here, we reported the time for deciding the existence of solution.

Contrary to the reprogramming, no other tools enable addressing this problem directly.
In addition to the CEGAR-2 and CEGAR-1, we also implemented CEGAR-0 which follows the same algorithm
but use no counter-example generalization, \ie,
$\raff(\vv{x})$ is $f \not = \vv f$.
This somehow corresponds to the enumeration and filtering used for marker reprogramming.
The aim is therefore to understand whether a CEGAR-based approach can be an effective first approach to attack the synthesis problem in accordance with universal marker properties.

\medskip
Instances of Moon dataset were run on a desktop computer with 
Intel(R) Xeon(R) E-2124 CPU at 3.30GHz and 64GB of RAM
; instances of Trappist were run on cluster nodes with AMD(R) Zen2
EPYC 7452 CPU at 2.35GHz and 256GB of RAM, all on a Linux operating system.

\subsection{Results}

\paragraph{Moon instances.}
Table~\ref{exp:moon} gives a comparison of execution times for the implementations
\begin{table}[p]
\caption{\label{exp:moon}%
Execution times (in sec.) for the reprogramming on the Moon dataset where
   $n$ is the number of components and $u$ the number of uncontrollable ones.
   Column ``First'' indicates the time for identifying the first solution, or the unsatisfiability of the
   instance (0 solutions); ``Enum'' the time for exhaustive solution enumeration; ``Solution'' the
   higher number of solution identified.
   When the problem turns out to have no solution, the time required is indicated in the column ``First''.
   For CEGAR methods, the number of identified counter-examples is in parentheses. 
    Then, TO(-) means that no counter-examples were found in 10 minutes.}
\resizebox{\textwidth}{!}{    \begin{tabular}{l|c|U|U|T|T|T|T|T|T|c|}
\cline{2-11}
                                              & k     & \multicolumn{2}{|c|}{\begin{tabular}[c]{@{}c@{}}Enum \& Filter \end{tabular}} & \multicolumn{2}{|c|}{\begin{tabular}[c]{@{}c@{}}Complementary\end{tabular}} &  \multicolumn{2}{|c|}{CEGAR-1} & \multicolumn{2}{|c|}{CEGAR-2} & Solutions                                                                           \\ \cline{3-10}
                                              
                                              \cline{3-10}
                                              &      & \multicolumn{1}{|c}{First} & \multicolumn{1}{|c|}{Enum} &   \multicolumn{1}{|c|}{First} & \multicolumn{1}{c|}{Enum} &  \multicolumn{1}{|c}{First} & \multicolumn{1}{|c|}{Enum} &  \multicolumn{1}{|c}{First} & \multicolumn{1}{|c|}{Enum}&                                                                            \\ \cline{1-11}
  
\multicolumn{1}{|l|}{{\begin{tabular}[c]{@{}l@{}}S1, Sahin et al. (2009)\\ n=20, u=1\end{tabular}}}            & 
\begin{tabular}[c]{@{}l@{}}2\\ 4\\ 6\end{tabular}&
\multicolumn{1}{|c}{\begin{tabular}[c]{@{}c@{}}0.03 \\ 0.03 \\ 0.04 \end{tabular}}  & 
\multicolumn{1}{|c|}{\begin{tabular}[c]{@{}c@{}}0.9 \\ 10.9 \\ TO\end{tabular}}  & 
\multicolumn{1}{|c|}{\begin{tabular}[c]{@{}c@{}}0.03 \\ 0.03 \\ \textbf{0.03} \end{tabular}}  & 
\multicolumn{1}{c|}{\begin{tabular}[c]{@{}c@{}}\textbf{0.1} \\ 1 \\ 1 \end{tabular}}  & 
\multicolumn{1}{|c}{\begin{tabular}[c]{@{}c@{}}0.07(4) \\ 0.8(21) \\ 1(22) \end{tabular}}  & 
\multicolumn{1}{|c|}{\begin{tabular}[c]{@{}c@{}}3.4(37) \\ 197.3(181) \\ TO(265) \end{tabular}}  &  
\multicolumn{1}{|c}{\begin{tabular}[c]{@{}c@{}}\textbf{0.02}(1) \\ \textbf{0.02}(1) \\ \textbf{0.03}(1)\end{tabular}}  & 
\multicolumn{1}{|c|}{\begin{tabular}[c]{@{}c@{}}\textbf{0.1}(2)\\ \textbf{0.1}(2) \\ \textbf{0.1}(2) \end{tabular}}  &
\begin{tabular}[c]{@{}c@{}} 9\\ 12\\ 12\end{tabular}    \\ \hline

\multicolumn{1}{|l|}{{\begin{tabular}[c]{@{}l@{}}S2, Wynn et al. (2012) \\ n=17, u=2\end{tabular}}}              &
\begin{tabular}[c]{@{}l@{}}2\\ 4\\ 6\end{tabular}&
\multicolumn{1}{|c}{\begin{tabular}[c]{@{}c@{}}0.6 \\ 0.5 \\ 0.7 \end{tabular}}  & 
\multicolumn{1}{|c|}{\begin{tabular}[c]{@{}c@{}}2.5 \\ 150.7 \\ TO\end{tabular}}  & 
\multicolumn{1}{|c|}{\begin{tabular}[c]{@{}c@{}}0.1 \\ \textbf{0.1} \\ \textbf{0.1} \end{tabular}}  & 
\multicolumn{1}{c|}{\begin{tabular}[c]{@{}c@{}}\textbf{0.1} \\ 4.4 \\ 20.7 \end{tabular}}  & 
\multicolumn{1}{|c}{\begin{tabular}[c]{@{}c@{}}0.6(21) \\ 0.6(22) \\ 0.5(20) \end{tabular}}  & 
\multicolumn{1}{|c|}{\begin{tabular}[c]{@{}c@{}}1.9(32) \\ TO(275) \\ TO(305) \end{tabular}}  &  
\multicolumn{1}{|c}{\begin{tabular}[c]{@{}c@{}}\textbf{0.04}(3) \\ \textbf{0.1}(3) \\ \textbf{0.1}(3)\end{tabular}}  & 
\multicolumn{1}{|c|}{\begin{tabular}[c]{@{}c@{}}\textbf{0.1}(4)\\ \textbf{0.2}(5) \\ \textbf{0.3}(4) \end{tabular}}  &
\begin{tabular}[c]{@{}c@{}}9\\ 19\\ 31\end{tabular}        \\ \hline

\multicolumn{1}{|l|}{{\begin{tabular}[c]{@{}l@{}}S3, Kasemeier-Kulesa\\  et al. (2018)\\n=18, u=4\end{tabular}}}  & 
\begin{tabular}[c]{@{}l@{}}2\\ 4\\ 6\end{tabular}&
\multicolumn{1}{|c}{\begin{tabular}[c]{@{}c@{}}2.3\\ 27.4 \\ 49.7 \end{tabular}}  & 
\multicolumn{1}{|c|}{\begin{tabular}[c]{@{}c@{}}- \\ 129.1 \\ TO\end{tabular}}  & 
\multicolumn{1}{|c|}{\begin{tabular}[c]{@{}c@{}}0.1\\ 3.6 \\ 3.6 \end{tabular}}  & 
\multicolumn{1}{c|}{\begin{tabular}[c]{@{}c@{}} -\\ 3.6 \\ 55.4 \end{tabular}}  & 
\multicolumn{1}{|c}{\begin{tabular}[c]{@{}c@{}}1.1(26) \\ 104.6(158) \\ TO(-)\end{tabular}}  & 
\multicolumn{1}{|c|}{\begin{tabular}[c]{@{}c@{}}- \\ TO(266) \\ TO(300) \end{tabular}}  &  
\multicolumn{1}{|c}{\begin{tabular}[c]{@{}c@{}}\textbf{0.03}(2) \\ \textbf{0.1}(3) \\ \textbf{0.1}(3)\end{tabular}}  & 
\multicolumn{1}{|c|}{\begin{tabular}[c]{@{}c@{}}-\\ \textbf{0.1}(3) \\ \textbf{0.2}(3) \end{tabular}}  &
\begin{tabular}[c]{@{}c@{}} 0\\ 12\\ 18\end{tabular}       \\ \hline

\multicolumn{1}{|l|}{{\begin{tabular}[c]{@{}l@{}}S4, Biane et al. (2019)\\n=13, u=2\end{tabular}}}    & 
\begin{tabular}[c]{@{}l@{}}2\\ 4\\ 6\end{tabular}&
\multicolumn{1}{|c}{\begin{tabular}[c]{@{}c@{}}\textbf{0.01} \\ 0.04 \\ 0.03 \end{tabular}}  & 
\multicolumn{1}{|c|}{\begin{tabular}[c]{@{}c@{}}0.6 \\ 0.5 \\ 0.5\end{tabular}}  & 
\multicolumn{1}{|c|}{\begin{tabular}[c]{@{}c@{}}0.02 \\ \textbf{0.02} \\ \textbf{0.02} \end{tabular}}  & 
\multicolumn{1}{c|}{\begin{tabular}[c]{@{}c@{}}\textbf{0.1} \\ \textbf{0.1} \\ \textbf{0.1} \end{tabular}}  & 
\multicolumn{1}{|c}{\begin{tabular}[c]{@{}c@{}}0.03(2) \\ 0.03(2) \\ 0.03(2) \end{tabular}}  & 
\multicolumn{1}{|c|}{\begin{tabular}[c]{@{}c@{}}0.2(10) \\ 0.3(13)\\ 0.3(13)\end{tabular}}  &  
\multicolumn{1}{|c}{\begin{tabular}[c]{@{}c@{}}0.02(1) \\ \textbf{0.02}(1) \\ \textbf{0.02}(1)\end{tabular}}  & 
\multicolumn{1}{|c|}{\begin{tabular}[c]{@{}c@{}}\textbf{0.1}(1)\\ \textbf{0.1}(1) \\ \textbf{0.1}(1) \end{tabular}}  &
\begin{tabular}[c]{@{}c@{}}9\\ 9\\ 9\end{tabular}   \\ \hline

\multicolumn{1}{|l|}{{\begin{tabular}[c]{@{}l@{}}M1, Calzone et al. (2010)\\n=28, u=3\end{tabular}}}          & 
\begin{tabular}[c]{@{}l@{}}2\\ 4\\ 6\end{tabular}&
\multicolumn{1}{|c}{\begin{tabular}[c]{@{}c@{}}0.3 \\ 0.3 \\ 0.2 \end{tabular}}  & 
\multicolumn{1}{|c|}{\begin{tabular}[c]{@{}c@{}}8.4 \\ TO \\ TO\end{tabular}}  & 
\multicolumn{1}{|c|}{\begin{tabular}[c]{@{}c@{}}\textbf{0.04} \\ \textbf{0.04} \\ \textbf{0.04} \end{tabular}}  & 
\multicolumn{1}{c|}{\begin{tabular}[c]{@{}c@{}}\textbf{0.3} \\ 49.3 \\ TO \end{tabular}}  & 
\multicolumn{1}{|c}{\begin{tabular}[c]{@{}c@{}}3.1(34) \\ 19.5(73)\\ 37.4(95) \end{tabular}}  & 
\multicolumn{1}{|c|}{\begin{tabular}[c]{@{}c@{}}32.5(71)\\ TO(221) \\ TO(245) \end{tabular}}  &  
\multicolumn{1}{|c}{\begin{tabular}[c]{@{}c@{}}0.1(3) \\ 0.1(3) \\ 0.1(3) \end{tabular}}  & 
\multicolumn{1}{|c|}{\begin{tabular}[c]{@{}c@{}}0.5(4)\\ \textbf{4}(9) \\ \textbf{7.9}(9) \end{tabular}}  &
\begin{tabular}[c]{@{}c@{}}36\\ 213\\ 370\end{tabular}    \\ \hline

\multicolumn{1}{|l|}{{\begin{tabular}[c]{@{}l@{}}M2 - Cohen et al. (2015)\\n=32, u=6\end{tabular}}}             & 
\begin{tabular}[c]{@{}l@{}}2\\ 4\\ 6\end{tabular}&
\multicolumn{1}{|c}{\begin{tabular}[c]{@{}c@{}}14.6 \\ 258 \\ 246.6 \end{tabular}}  & 
\multicolumn{1}{|c|}{\begin{tabular}[c]{@{}c@{}}- \\ TO \\ TO\end{tabular}}  & 
\multicolumn{1}{|c|}{\begin{tabular}[c]{@{}c@{}}0.6 \\ 6.7 \\ 6.7 \end{tabular}}  & 
\multicolumn{1}{c|}{\begin{tabular}[c]{@{}c@{}} -\\ 78.1 \\ TO \end{tabular}}  & 
\multicolumn{1}{|c}{\begin{tabular}[c]{@{}c@{}}\textbf{0.03}(1) \\ 1(14)\\ \textbf{0.02}(0) \end{tabular}}  & 
\multicolumn{1}{|c|}{\begin{tabular}[c]{@{}c@{}}- \\ 43.2(61) \\ TO(182) \end{tabular}}  &  
\multicolumn{1}{|c}{\begin{tabular}[c]{@{}c@{}}\textbf{0.03}(1) \\ \textbf{0.1}(2) \\ \textbf{0.02}(0)\end{tabular}}  & 
\multicolumn{1}{|c|}{\begin{tabular}[c]{@{}c@{}}-\\ \textbf{0.6}(7) \\ \textbf{2.3}(8) \end{tabular}}  &
\begin{tabular}[c]{@{}c@{}} 0\\ 14\\ 78\end{tabular}  \\ \hline

\multicolumn{1}{|l|}{{\begin{tabular}[c]{@{}l@{}}M3 - Remy et al. (2015)\\n=35, u=4\end{tabular}}}             & 
\begin{tabular}[c]{@{}l@{}}2\\ 4\\ 6\end{tabular}&
\multicolumn{1}{|c}{\begin{tabular}[c]{@{}c@{}}4.6 \\ 4 \\ 6.4 \end{tabular}}  & 
\multicolumn{1}{|c|}{\begin{tabular}[c]{@{}c@{}}21.1 \\ TO \\ TO\end{tabular}}  & 
\multicolumn{1}{|c|}{\begin{tabular}[c]{@{}c@{}}0.8 \\ 0.9 \\ 0.9 \end{tabular}}  & 
\multicolumn{1}{c|}{\begin{tabular}[c]{@{}c@{}} 0.8\\ 189.6 \\ TO \end{tabular}}  & 
\multicolumn{1}{|c}{\begin{tabular}[c]{@{}c@{}}255.2(116) \\ 8.6(38)\\ 8.1(37) \end{tabular}}  & 
\multicolumn{1}{|c|}{\begin{tabular}[c]{@{}c@{}} 334.4(124) \\ TO(173) \\ TO(180) \end{tabular}}  &  
\multicolumn{1}{|c}{\begin{tabular}[c]{@{}c@{}}\textbf{0.2}(5) \\ \textbf{0.1}(3) \\ \textbf{0.1}(3)\end{tabular}}  & 
\multicolumn{1}{|c|}{\begin{tabular}[c]{@{}c@{}}\textbf{0.3}(5)\\ \textbf{8.5}(11) \\ \textbf{235.9}(17) \end{tabular}}  &
\begin{tabular}[c]{@{}c@{}} 3\\ 261\\ 2015\end{tabular}  \\ \hline

\multicolumn{1}{|l|}{\begin{tabular}[c]{@{}l@{}}L1, Saadatpour et al.\\ (2011)\\n=59, u=3\end{tabular}}         & 
\begin{tabular}[c]{@{}l@{}}2\\ 4\\ 6\end{tabular}&
\multicolumn{1}{|c}{\begin{tabular}[c]{@{}c@{}}98.2 \\ TO \\ TO \end{tabular}}  & 
\multicolumn{1}{|c|}{\begin{tabular}[c]{@{}c@{}} -\\ TO \\ TO\end{tabular}}  & 
\multicolumn{1}{|c|}{\begin{tabular}[c]{@{}c@{}}3.7 \\ 119.9 \\ 120.6 \end{tabular}}  & 
\multicolumn{1}{c|}{\begin{tabular}[c]{@{}c@{}} -\\ TO \\ TO \end{tabular}}  & 
\multicolumn{1}{|c}{\begin{tabular}[c]{@{}c@{}}0.8(11) \\ TO(-) \\ TO(-) \end{tabular}}  & 
\multicolumn{1}{|c|}{\begin{tabular}[c]{@{}c@{}} -\\ TO(153) \\ TO(170) \end{tabular}}  &  
\multicolumn{1}{|c}{\begin{tabular}[c]{@{}c@{}}\textbf{0.04}(1) \\ \textbf{0.1}(2) \\ \textbf{0.1}(3)\end{tabular}}  & 
\multicolumn{1}{|c|}{\begin{tabular}[c]{@{}c@{}}-\\ \textbf{2.8}(5) \\ TO(21) \end{tabular}}  &
\begin{tabular}[c]{@{}c@{}} 0\\ 83\\ $\geq$ 2227 \end{tabular}\\ \hline

\multicolumn{1}{|l|}{\begin{tabular}[c]{@{}l@{}}L2, Singh et al. (2012)\\n=66, u=1\end{tabular}}            &
\begin{tabular}[c]{@{}l@{}}2\\ 4\\ 6\end{tabular}&
\multicolumn{1}{|c}{\begin{tabular}[c]{@{}c@{}}\textbf{0.1} \\ \textbf{0.1} \\ \textbf{0.1} \end{tabular}}  & 
\multicolumn{1}{|c|}{\begin{tabular}[c]{@{}c@{}}75 \\ TO \\ TO\end{tabular}}  & 
\multicolumn{1}{|c|}{\begin{tabular}[c]{@{}c@{}}0.2 \\ 0.2 \\ 0.2 \end{tabular}}  & 
\multicolumn{1}{c|}{\begin{tabular}[c]{@{}c@{}}3.4 \\ TO \\ TO \end{tabular}}  & 
\multicolumn{1}{|c}{\begin{tabular}[c]{@{}c@{}}0.3(6) \\ 0.3(6) \\ 0.3(6) \end{tabular}}  & 
\multicolumn{1}{|c|}{\begin{tabular}[c]{@{}c@{}}TO(184) \\ TO(187) \\ TO(189) \end{tabular}}  &  
\multicolumn{1}{|c}{\begin{tabular}[c]{@{}c@{}}\textbf{0.1}(1) \\ \textbf{0.1}(1) \\ \textbf{0.1}(1)\end{tabular}}  & 
\multicolumn{1}{|c|}{\begin{tabular}[c]{@{}c@{}}\textbf{1.2}(1)\\ \textbf{1.2}(1) \\ \textbf{1.2}(1) \end{tabular}}  &
\begin{tabular}[c]{@{}c@{}} 60\\ 60\\ 60\end{tabular}      \\ \hline

\multicolumn{1}{|l|}{\begin{tabular}[c]{@{}l@{}}L3, Grieco et al. (2013)\\n=53, u=3\end{tabular}}            & 
\begin{tabular}[c]{@{}l@{}}2\\ 4\\ 6\end{tabular}&
\multicolumn{1}{|c}{\begin{tabular}[c]{@{}c@{}}24.3 \\ 8.9 \\ 18.3 \end{tabular}}  & 
\multicolumn{1}{|c|}{\begin{tabular}[c]{@{}c@{}}59.8 \\ TO \\ TO\end{tabular}}  & 
\multicolumn{1}{|c|}{\begin{tabular}[c]{@{}c@{}}2.4 \\ 2.4\\ 2.5 \end{tabular}}  & 
\multicolumn{1}{c|}{\begin{tabular}[c]{@{}c@{}}2.4 \\ TO \\ TO \end{tabular}}  & 
\multicolumn{1}{|c}{\begin{tabular}[c]{@{}c@{}}\textbf{0.1}(1) \\ \textbf{0.1}(1) \\ \textbf{0.1}(2) \end{tabular}}  & 
\multicolumn{1}{|c|}{\begin{tabular}[c]{@{}c@{}}TO(159) \\ TO(171) \\ TO(177) \end{tabular}}  &  
\multicolumn{1}{|c}{\begin{tabular}[c]{@{}c@{}}\textbf{0.1}(1) \\ \textbf{0.1}(1) \\ \textbf{0.1}(1)\end{tabular}}  & 
\multicolumn{1}{|c|}{\begin{tabular}[c]{@{}c@{}}\textbf{0.2}(2)\\ \textbf{77}(23)\\ TO(33) \end{tabular}}  &
\begin{tabular}[c]{@{}c@{}} 8\\ 722\\ $\geq$ 2171\end{tabular} \\ \hline

\multicolumn{1}{|l|}{\begin{tabular}[c]{@{}l@{}}L4, Flobak et al. (2015)\\n=75, u=2\end{tabular}}             & 
\begin{tabular}[c]{@{}l@{}}2\\ 4\\ 6\end{tabular}&
\multicolumn{1}{|c}{\begin{tabular}[c]{@{}c@{}}153.1 \\ TO \\ TO \end{tabular}}  & 
\multicolumn{1}{|c|}{\begin{tabular}[c]{@{}c@{}}- \\ TO \\ TO\end{tabular}}  & 
\multicolumn{1}{|c|}{\begin{tabular}[c]{@{}c@{}}6.2 \\ 228.1 \\ 227.8 \end{tabular}}  & 
\multicolumn{1}{c|}{\begin{tabular}[c]{@{}c@{}} -\\ TO \\ TO \end{tabular}}  & 
\multicolumn{1}{|c}{\begin{tabular}[c]{@{}c@{}}\textbf{0.02}(0) \\ \textbf{0.03}(0) \\ \textbf{0.04}(0) \end{tabular}}  & 
\multicolumn{1}{|c|}{\begin{tabular}[c]{@{}c@{}} -\\ TO(112)\\ TO(179) \end{tabular}}  &  
\multicolumn{1}{|c}{\begin{tabular}[c]{@{}c@{}}\textbf{0.02}(0) \\ \textbf{0.03}(0) \\ \textbf{0.04}(0)\end{tabular}}  & 
\multicolumn{1}{|c|}{\begin{tabular}[c]{@{}c@{}}-\\ \textbf{78.2}(1) \\ \textbf{518}(5) \end{tabular}}  &
\begin{tabular}[c]{@{}c@{}} 0\\ 1302\\ 3435\end{tabular}   \\ \hline
\end{tabular}
}
\end{table}
\begin{table}[p]
\caption{\label{exp:syn_moon}%
Execution times (in sec.) for determining the satisfiability of the synthesis on the Moon dataset with a 10 minutes timeout (TO); $\dagger$ indicates unsatisfiable instances.
Number of identified counter-examples is in parentheses.}
\resizebox{\textwidth}{!}{\begin{tabular}{@{}c|c|c|c|c|c|c|c|c|c|c|c|@{}}
\cline{2-12}
                                       & S1 (20)                                                & S2 (17)                                                & S3 (18)                                                  & S4 (13)                                               & M1 (28)                                                & M2 (32)     
                                       & M3 (35)                                          & L1 (59)                                                & L2 (66)                                                  & L3 (53)                                               & L4 (75)                                               \\ \hline

\multicolumn{1}{|c|}{CEGAR-0} &
0.2(0) & TO(59,584) & TO(53,300) & TO(63,914) &TO(44,585) & TO(4,451) &TO(9,987) & TO(6,596) & TO(24,686) &  0.2(0)& 1.2(0) \\

\multicolumn{1}{|c|}{CEGAR-1} &
\textbf{0.1}(0)& 1(20)$\dagger$ & 12.4(49)$\dagger$ & 0.04(1) & \textbf{0.2}(4) & TO(5) &249.3(5) &TO(36) & \textbf{0.1}(2)$\dagger$ & \textbf{0.2}(0) & 1.3(0)\\ 

\multicolumn{1}{|c|}{CEGAR-2} &
\textbf{0.1}(0) & \textbf{0.1}(3)$\dagger$ & \textbf{0.1}(3)$\dagger$ & \textbf{0.03}(1) & \textbf{0.2}(6) & \textbf{97.8}(1) & \textbf{84.2}(3)& \textbf{2.4}(4) & \textbf{0.1}(1)$\dagger$ & \textbf{0.2}(0) & \textbf{1.2}(0)\\ \hline

\end{tabular}
}
\end{table}
of the marker reprogramming, depending on the maximum size $k$ of perturbations to identify.
CEGAR-2 is the only one able to always determine the satisfiability of the instances 
in 10 minutes.
Moreover, it largely outperforms other methods for the enumeration as soon as $k$ or $n$ is large.
Through the results obtained with CEGAR-1 and CEGAR-2 on the datasets, we can conclude that 
imposing the trap spaces of counter-example
configurations to contain a trap spaces matching with the marker helps to drastically reduce the number
of iterations to exhaust the solution space.
Table~\ref{exp:syn_moon} provides a similar picture for the synthesis problem.
Note that CEGAR-0 only solved instances where the first generated BN verified the universal property.

\paragraph{Trappist instances.}
Fig.~\ref{fig:trap3} and \ref{fig:trapsyn} provide summary statistics of CEGAR-2 performance for the
reprogramming and synthesis, grouped by network size.
The execution time is limited to 30 minutes for reprogramming and 60 minutes for synthesis.
Full results are provided in Appendix~\ref{sec:results-trappist}.
The dataset consists of much larger networks than in Moon, although the employed markers may not be
biologically meaningful.
For the reprogramming, our prototype always decides whether the instance admits a solution
for all networks up to 400 components.
Failed exhaustive enumerations were then likely caused by a large combinatorics of solutions with large $k$.
For BNs above 1,000 components, our prototype managed to determine the satisfiability of 75\% of the
instances within the given time limit.
In the case of the synthesis, our prototype was able to solve more than 80\% of the instances of the
BNs below 400 components, and 50\% of networks above 1,000.

\begin{figure}[tb]
    \centering
    \includegraphics[width=\textwidth]{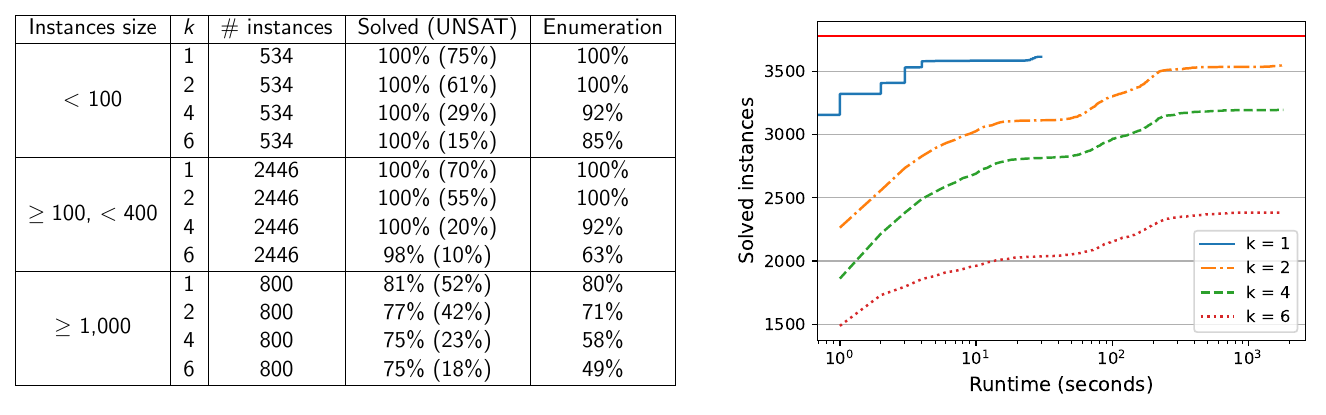}
    \vspace{-1cm}
    \caption{\label{fig:trap3}
        Summary of results for the reprogramming of Trappist instances.
        (left) ratio of instances for which the satisfiability has been determined with
        relative ratio of unsatisfiable, and for which the exhaustive
        enumeration has been completed within 30min.
    (right) number of instances with exhaustive enumeration completed within given time. 
    The red line indicates the total number of instances}
\end{figure}
\begin{figure}[tb]
    \centering
    \includegraphics[width=\textwidth]{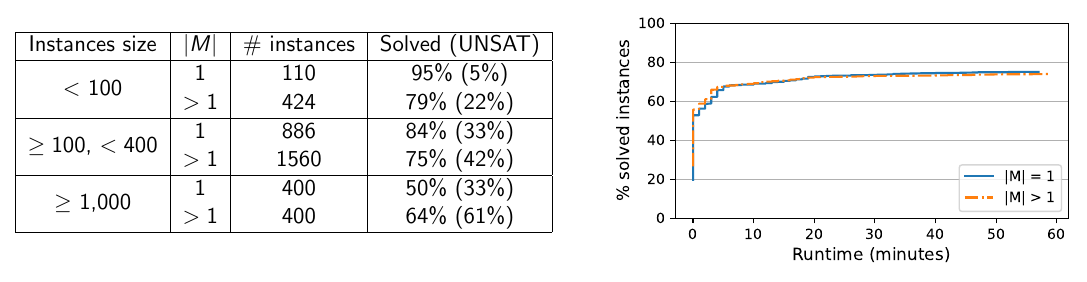}
    \vspace{-1cm}
    \caption{\label{fig:trapsyn}
    Summary of results for the synthesis of Trappist instances within a 1h limit}
\end{figure}

\section{Discussion}

We demonstrated a new approach to efficiently reason on universal properties over
minimal trap spaces of BNs, by iterative refinements of a logical satisfiability problem guided by
counter-examples.
Our prototype scaled to BNs with thousands of components for solving reprogramming and
synthesis problems.

In the experiments, it appeared that only a limited number of counter-examples are necessary to
exhaust the solution space of reprogramming instances and for deciding the satisfiability of the
synthesis instances.
Indeed, 90\% of the trappist instances required less than 20
counter-example configurations, the maximum being 130 counter-example configurations for one
case.
It is very low compared to the number of configurations of considered networks, and actually always
lower than the number of components.
Thus, in the conducted experiments, the 3-QBF problems have been resolved with a (sub)linear number
(with the dimension of the BN) of 2-QBF problems.
Theoretical and experimental work may further investigate the connection between the number of
counter-examples and the diversity of minimal trap spaces that can be generated by the candidate
perturbations and Boolean functions.
This could offer a finer insight on the practical complexity of tackled problems.

Our method can be employed to solve more general reprogramming and synthesis problems, for instance
for enforcing the absence of cyclic attractors,
or universal marker properties over attractors reachable in the most permissive dynamics from a given
configuration.
We plan to embed this generic solving technique in the \textsc{BoNesis} software, promising a
scalable BN synthesis from rich and biologically-relevant dynamical properties.

\subsubsection{Acknowledgements}
SR, GML and LP
We acknowledge support of the French Agence Nationale pour la Recherche (ANR)
in the scope of the project ``\href{https://bnediction.github.io}{BNeDiction}'' (ANR-20-CE45-0001)
and of the project ``PING/ACK'' (ANR-18-CE40-0011).
Part of the experiments presented in this paper were carried out using the PlaFRIM experimental testbed,
supported by Inria, CNRS (LABRI and IMB), Université de Bordeaux, Bordeaux INP and Conseil Régional
d’Aquitaine (see \url{https://www.plafrim.fr}).

\bibliographystyle{abbrvurl}
\bibliography{biblio}

\begin{thebibliography}{10}

\bibitem{abio2016cnf}
I.~Ab{\'\i}o, G.~Gange, V.~Mayer-Eichberger, and P.~J. Stuckey.
\newblock On cnf encodings of decision diagrams.
\newblock In {\em Integration of AI and OR Techniques in Constraint
  Programming: 13th International Conference, CPAIOR 2016, Banff, AB, Canada,
  May 29-June 1, 2016, Proceedings 13}, pages 1--17. Springer, 2016.
\newblock \href {https://doi.org/10.1007/978-3-319-33954-2_1}
  {\path{doi:10.1007/978-3-319-33954-2_1}}.

\bibitem{aracena2008maximum}
J.~Aracena.
\newblock Maximum number of fixed points in regulatory boolean networks.
\newblock {\em Bulletin of mathematical biology}, 70:1398--1409, 2008.
\newblock \href {https://doi.org/10.1007/s11538-008-9304-7}
  {\path{doi:10.1007/s11538-008-9304-7}}.

\bibitem{aracena2004positive}
J.~Aracena, J.~Demongeot, and E.~Goles.
\newblock Positive and negative circuits in discrete neural networks.
\newblock {\em IEEE Transactions on Neural Networks}, 15(1):77--83, 2004.
\newblock \href {https://doi.org/10.1109/TNN.2003.821555}
  {\path{doi:10.1109/TNN.2003.821555}}.

\bibitem{baral2003knowledge}
C.~Baral.
\newblock {\em Knowledge representation, reasoning and declarative problem
  solving}.
\newblock Cambridge university press, 2003.
\newblock \href {https://doi.org/10.1017/CBO9780511543357}
  {\path{doi:10.1017/CBO9780511543357}}.

\bibitem{Benes2020}
N.~Bene{\v{s}}, L.~Brim, J.~Kadlecaj, S.~Pastva, and D.~{\v{S}}afr{\'{a}}nek.
\newblock {AEON}: Attractor bifurcation analysis of parametrised boolean
  networks.
\newblock In {\em Computer Aided Verification}, pages 569--581. Springer
  International Publishing, 2020.
\newblock \href {https://doi.org/10.1007/978-3-030-53288-8_28}
  {\path{doi:10.1007/978-3-030-53288-8_28}}.

\bibitem{biane2018causal}
C.~Biane and F.~Delaplace.
\newblock Causal reasoning on boolean control networks based on abduction:
  theory and application to cancer drug discovery.
\newblock {\em IEEE/ACM transactions on computational biology and
  bioinformatics}, 16(5):1574--1585, 2018.
\newblock \href {https://doi.org/10.1109/tcbb.2018.2889102}
  {\path{doi:10.1109/tcbb.2018.2889102}}.

\bibitem{booksat}
A.~Biere, M.~Heule, H.~van Maaren, and T.~Walsh, editors.
\newblock {\em Handbook of Satisfiability - Second Edition}, volume 336 of {\em
  Frontiers in Artificial Intelligence and Applications}.
\newblock {IOS} Press, 2021.
\newblock \href {https://doi.org/10.3233/FAIA336} {\path{doi:10.3233/FAIA336}}.

\bibitem{propagationref}
L.~Bordeaux and J.~Marques-Silva.
\newblock Knowledge compilation with empowerment.
\newblock In {\em SOFSEM 2012: Theory and Practice of Computer Science}, pages
  612--624, Berlin, Heidelberg, 2012. Springer.
\newblock \href {https://doi.org/10.1007/978-3-642-27660-6_50}
  {\path{doi:10.1007/978-3-642-27660-6_50}}.

\bibitem{bordeaux2012knowledge}
L.~Bordeaux and J.~Marques-Silva.
\newblock Knowledge compilation with empowerment.
\newblock In {\em SOFSEM 2012: Theory and Practice of Computer Science}, pages
  612--624. Springer, 2012.
\newblock \href {https://doi.org/10.1007/978-3-642-27660-6_50}
  {\path{doi:10.1007/978-3-642-27660-6_50}}.

\bibitem{Bning2021TheoryOQ}
H.~K. B{\"u}ning and U.~Bubeck.
\newblock Theory of quantified boolean formulas.
\newblock In {\em Handbook of Satisfiability}, 2021.
\newblock \href {https://doi.org/10.3233/978-1-58603-929-5-735}
  {\path{doi:10.3233/978-1-58603-929-5-735}}.

\bibitem{chevalier2019synthesis}
S.~Chevalier, C.~Froidevaux, L.~Paulev{\'e}, and A.~Zinovyev.
\newblock Synthesis of boolean networks from biological dynamical constraints
  using answer-set programming.
\newblock In {\em 2019 IEEE 31st International Conference on Tools with
  Artificial Intelligence (ICTAI)}, pages 34--41. IEEE, 2019.
\newblock \href {https://doi.org/10.1109/ICTAI.2019.00014}
  {\path{doi:10.1109/ICTAI.2019.00014}}.

\bibitem{Chevalier2020-CMSB}
S.~Chevalier, V.~No{\"e}l, L.~Calzone, A.~Zinovyev, and L.~Paulev{\'e}.
\newblock {Synthesis and Simulation of Ensembles of Boolean Networks for Cell
  Fate Decision}.
\newblock In {\em {CMSB 2020 - 18th International Conference on Computational
  Methods in Systems Biology}}, volume 12314 of {\em Lecture Notes in Computer
  Science}, pages 193--209, Cham, 2020. Springer.
\newblock \href {https://doi.org/10.1007/978-3-030-60327-4\_11}
  {\path{doi:10.1007/978-3-030-60327-4\_11}}.

\bibitem{cifuentes2022control}
L.~Cifuentes-Fontanals, E.~Tonello, and H.~Siebert.
\newblock Control in boolean networks with model checking.
\newblock {\em Frontiers in Applied Mathematics and Statistics}, 8:838546,
  2022.
\newblock \href {https://doi.org/10.3389/fams.2022.838546}
  {\path{doi:10.3389/fams.2022.838546}}.

\bibitem{clarke2003counterexample}
E.~Clarke, O.~Grumberg, S.~Jha, Y.~Lu, and H.~Veith.
\newblock Counterexample-guided abstraction refinement for symbolic model
  checking.
\newblock {\em Journal of the ACM (JACM)}, 50(5):752--794, 2003.
\newblock \href {https://doi.org/10.1145/876638.876643}
  {\path{doi:10.1145/876638.876643}}.

\bibitem{dorier2016boolean}
J.~Dorier, I.~Crespo, A.~Niknejad, R.~Liechti, M.~Ebeling, and I.~Xenarios.
\newblock Boolean regulatory network reconstruction using literature based
  knowledge with a genetic algorithm optimization method.
\newblock {\em BMC bioinformatics}, 17:1--19, 2016.
\newblock \href {https://doi.org/10.1186/s12859-016-1287-z}
  {\path{doi:10.1186/s12859-016-1287-z}}.

\bibitem{drechsler2013binary}
R.~Drechsler and B.~Becker.
\newblock {\em Binary decision diagrams: theory and implementation}.
\newblock Springer Science \& Business Media, 2013.
\newblock \href {https://doi.org/10.1007/978-1-4757-2892-7}
  {\path{doi:10.1007/978-1-4757-2892-7}}.

\bibitem{frisch2010sat}
A.~M. Frisch and P.~A. Giannaros.
\newblock Sat encodings of the at-most-k constraint: Some old, some new, some
  fast, some slow.
\newblock 2010.
\newblock \href {https://doi.org/10.1007/978-3-030-30446-1_7}
  {\path{doi:10.1007/978-3-030-30446-1_7}}.

\bibitem{frohlich2014idq}
A.~Fr{\"o}hlich, G.~Kov{\'a}sznai, A.~Biere, and H.~Veith.
\newblock idq: Instantiation-based dqbf solving.
\newblock In {\em POS@ SAT}, pages 103--116, 2014.
\newblock \href {https://doi.org/10.29007/1s5k} {\path{doi:10.29007/1s5k}}.

\bibitem{gebser2012answer}
M.~Gebser, R.~Kaminski, B.~Kaufmann, and T.~Schaub.
\newblock Answer set solving in practice.
\newblock {\em Synthesis lectures on artificial intelligence and machine
  learning}, 6(3):1--238, 2012.
\newblock \href {https://doi.org/10.1007/978-3-031-01561-8}
  {\path{doi:10.1007/978-3-031-01561-8}}.

\bibitem{Gebser2018}
M.~Gebser, R.~Kaminski, B.~Kaufmann, and T.~Schaub.
\newblock Multi-shot {ASP} solving with clingo.
\newblock {\em Theory and Practice of Logic Programming}, 19(1):27--82, 2018.
\newblock \href {https://doi.org/10.1017/s1471068418000054}
  {\path{doi:10.1017/s1471068418000054}}.

\bibitem{Goldfeder2019}
J.~Goldfeder and H.~Kugler.
\newblock {BRE}:{IN} - a backend for reasoning about interaction networks with
  temporal logic.
\newblock In {\em Computational Methods in Systems Biology}, pages 289--295.
  Springer International Publishing, 2019.
\newblock \href {https://doi.org/10.1007/978-3-030-31304-3_15}
  {\path{doi:10.1007/978-3-030-31304-3_15}}.

\bibitem{klarner2018basins}
H.~Klarner, F.~Heinitz, S.~Nee, and H.~Siebert.
\newblock Basins of attraction, commitment sets, and phenotypes of boolean
  networks.
\newblock {\em IEEE/ACM transactions on computational biology and
  bioinformatics}, 17(4):1115--1124, 2018.
\newblock \href {https://doi.org/10.1109/TCBB.2018.2879097}
  {\path{doi:10.1109/TCBB.2018.2879097}}.

\bibitem{klarner2015approximating}
H.~Klarner and H.~Siebert.
\newblock Approximating attractors of boolean networks by iterative ctl model
  checking.
\newblock {\em Frontiers in bioengineering and biotechnology}, 3:130, 2015.
\newblock \href {https://doi.org/10.3389/fbioe.2015.00130}
  {\path{doi:10.3389/fbioe.2015.00130}}.

\bibitem{lonsing2017depqbf}
F.~Lonsing and U.~Egly.
\newblock Depqbf 6.0: A search-based qbf solver beyond traditional qcdcl.
\newblock In {\em Automated Deduction--CADE 26: 26th International Conference
  on Automated Deduction, Gothenburg, Sweden, August 6--11, 2017, Proceedings},
  pages 371--384. Springer, 2017.
\newblock \href {https://doi.org/10.1007/978-3-319-63046-5_23}
  {\path{doi:10.1007/978-3-319-63046-5_23}}.

\bibitem{montagud2022patient}
A.~Montagud, J.~B{\'e}al, L.~Tobalina, P.~Traynard, V.~Subramanian, B.~Szalai,
  R.~Alf{\"o}ldi, L.~Pusk{\'a}s, A.~Valencia, E.~Barillot, et~al.
\newblock Patient-specific boolean models of signalling networks guide
  personalised treatments.
\newblock {\em Elife}, 11:e72626, 2022.
\newblock \href {https://doi.org/10.7554/eLife.72626}
  {\path{doi:10.7554/eLife.72626}}.

\bibitem{montagud22}
A.~Montagud, J.~B{\'{e}}al, L.~Tobalina, P.~Traynard, V.~Subramanian,
  B.~Szalai, R.~Alföldi, L.~Pusk{\'{a}}s, A.~Valencia, E.~Barillot,
  J.~Saez-Rodriguez, and L.~Calzone.
\newblock Patient-specific boolean models of signalling networks guide
  personalised treatments.
\newblock {\em {eLife}}, 11, 2022.
\newblock \href {https://doi.org/10.7554/elife.72626}
  {\path{doi:10.7554/elife.72626}}.

\bibitem{moon2022bilevel}
K.~Moon, K.~Lee, S.~Chopra, and S.~Kwon.
\newblock Bilevel integer programming on a boolean network for discovering
  critical genetic alterations in cancer development and therapy.
\newblock {\em European Journal of Operational Research}, 300(2):743--754,
  2022.
\newblock \href {https://doi.org/10.1016/j.ejor.2021.10.019}
  {\path{doi:10.1016/j.ejor.2021.10.019}}.

\bibitem{moon2022computational}
K.~Moon, K.~Lee, and L.~Paulev{\'e}.
\newblock Computational complexity of minimal trap spaces in boolean networks.
\newblock {\em ArXiv e-prints}, 2022.
\newblock \href {https://doi.org/10.48550/arXiv.2212.12756}
  {\path{doi:10.48550/arXiv.2212.12756}}.

\bibitem{pauleve2022marker}
L.~Paulev{\'e}.
\newblock Marker and source-marker reprogramming of {Most} {Permissive}
  {Boolean} networks and ensembles with {BoNesis }.
\newblock {\em Peer Community Journal}, 3, 2023.
\newblock \href {https://doi.org/10.24072/pcjournal.255}
  {\path{doi:10.24072/pcjournal.255}}.

\bibitem{naturepauleve}
L.~Paulev{\'e}, J.~Kol{\v c}{\'a}k, T.~Chatain, and S.~Haar.
\newblock Reconciling qualitative, abstract, and scalable modeling of
  biological networks.
\newblock {\em Nature Communications}, 11(1):4256, 2020.
\newblock \href {https://doi.org/10.1038/s41467-020-18112-5}
  {\path{doi:10.1038/s41467-020-18112-5}}.

\bibitem{automata21}
L.~Paulev{\'e} and S.~Sen{\'e}.
\newblock {Non-deterministic updates of Boolean networks}.
\newblock In {\em 27th IFIP WG 1.5 International Workshop on Cellular Automata
  and Discrete Complex Systems (AUTOMATA 2021)}, volume~90 of {\em Open Access
  Series in Informatics (OASIcs)}, pages 10:1--10:16. Schloss Dagstuhl --
  Leibniz-Zentrum f{\"u}r Informatik, 2021.
\newblock \href {https://doi.org/10.4230/OASIcs.AUTOMATA.2021.10}
  {\path{doi:10.4230/OASIcs.AUTOMATA.2021.10}}.

\bibitem{pauleve2022boolean}
L.~Paulev{\'e} and S.~Sen{\'e}.
\newblock {Boolean networks and their dynamics: the impact of updates}.
\newblock In {\em {Systems Biology Modelling and Analysis: Formal
  Bioinformatics Methods and Tools}}. {Wiley}, 2022.
\newblock \href {https://doi.org/10.1002/9781119716600.ch6}
  {\path{doi:10.1002/9781119716600.ch6}}.

\bibitem{rabe2015caqe}
M.~N. Rabe and L.~Tentrup.
\newblock Caqe: a certifying qbf solver.
\newblock In {\em 2015 Formal Methods in Computer-Aided Design (FMCAD)}, pages
  136--143. IEEE, 2015.
\newblock \href {https://doi.org/10.1109/FMCAD.2015.7542263}
  {\path{doi:10.1109/FMCAD.2015.7542263}}.

\bibitem{Reda2022}
C.~R{\'{e}}da and A.~Delahaye-Duriez.
\newblock Prioritization of~candidate genes through boolean networks.
\newblock In {\em Computational Methods in Systems Biology}, pages 89--121.
  Springer International Publishing, 2022.
\newblock \href {https://doi.org/10.1007/978-3-031-15034-0_5}
  {\path{doi:10.1007/978-3-031-15034-0_5}}.

\bibitem{dqbdd}
J.~S{\'i}{\v{c}} and J.~Strej{\v{c}}ek.
\newblock Dqbdd: An efficient bdd-based dqbf solver.
\newblock In {\em Theory and Applications of Satisfiability Testing -- SAT
  2021}, pages 535--544, Cham, 2021. Springer International Publishing.
\newblock \href {https://doi.org/10.1007/978-3-030-80223-3_36}
  {\path{doi:10.1007/978-3-030-80223-3_36}}.

\bibitem{trinh2022minimal}
V.-G. Trinh, B.~Benhamou, K.~Hiraishi, and S.~Soliman.
\newblock Minimal trap spaces of logical models are maximal siphons of their
  petri net encoding.
\newblock In {\em Computational Methods in Systems Biology: 20th International
  Conference, CMSB 2022, Bucharest, Romania, September 14--16, 2022,
  Proceedings}, pages 158--176. Springer, 2022.
\newblock \href {https://doi.org/10.1007/978-3-031-15034-0_8}
  {\path{doi:10.1007/978-3-031-15034-0_8}}.

\bibitem{yordanov2016method}
B.~Yordanov, S.-J. Dunn, H.~Kugler, A.~Smith, G.~Martello, and S.~Emmott.
\newblock A method to identify and analyze biological programs through
  automated reasoning.
\newblock {\em NPJ systems biology and applications}, 2(1):1--16, 2016.
\newblock \href {https://doi.org/10.1038/npjsba.2016.10}
  {\path{doi:10.1038/npjsba.2016.10}}.

\bibitem{zanudo21}
J.~G.~T. Za{\~{n}}udo, P.~Mao, C.~Alcon, K.~Kowalski, G.~N. Johnson, G.~Xu,
  J.~Baselga, M.~Scaltriti, A.~Letai, J.~Montero, R.~Albert, and N.~Wagle.
\newblock Cell line-specific network models of {ER}+ breast cancer identify
  potential {PI}3ka inhibitor resistance mechanisms and drug combinations.
\newblock {\em Cancer Research}, 81(17):4603--4617, 2021.
\newblock \href {https://doi.org/10.1158/0008-5472.can-21-1208}
  {\path{doi:10.1158/0008-5472.can-21-1208}}.

\end{thebibliography}

\appendix

\section{Encoding of Trap Space in propositional logic}\label{sec:ts-cnf}
Following the explanation of the CEGAR approach, for reprogramming, we can see that a key element, in being able to apply this idea, is the computation of the $TS_f(x)$. 
For this reason, we will present how it is possible to encode the problem in a Boolean formula.
After a general explanation, we will briefly present how this idea can be applied both for locally monotone BNs and, for example, for all BNs with all $f_i$ given in a propagation complete representation. 

\smallskip

As explained in Section \ref{sec:complexity}, a possible approach to compute $TS_f(x)$ can be to start from $TS_f(x)=x$ and then, considering one component at a time, verify if there exists a $z\in TS_f(x))$ such that $f_i(z)\neq (TS_f(x))_i$.  
In this last case, $(TS_f(x))_i$ becomes $\ast$ and it will remain a $\ast$ forever. 
To obtain the smallest trap space containing the configuration, the process must be iterated until $TS_f(x)$ cannot be changed anymore and the closeness property is then satisfied.

\begin{example}\label{ex:ts}
Let us consider the BN $f_1(x)=x_2$, $f_2(x)=x_3 \land x_4$, $f_3(x)=x_4 \land \neg x_2$ and $f_4(x)=\neg x_1 \lor x_4$. 
To compute $TS_f(0000)$, let us start from $TS_f(0000)=0000$ (with $0000)=\set{0000}$). 
Considering one component at a time (performing the first iteration), we obtain $f_1(0000)=0$, $f_2(0000)=0$, and  $f_3(0000)=0$ and $f_4(0000)=1$ which implies $TS_f(0000)=000\ast$ and $000\ast)=\set{0001,0000}$. 
Performing the second iteration, we discover that $\nexists z\in 000\ast)$ such that $f_1(z)$ or $f_2(z)$ diffear from $0$, but $f_3(0001)=1$.
Then, $TS_f(0000)=00\ast\ast$ and $00\ast\ast)=\set{0011,0010,0001,0000}$.
In the third iteration, $f_2(0011)=1$ brings $TS_f(0000)=0\ast^3$.
Finally, in the last iteration, $f_1(0111)=1$ implies $TS_f(0000)=\ast^4$.
\end{example}

It is clear that $TS_f(x)$ is computed with different iterations over the entire configuration until it is impossible to change a $h_i$ in a $\ast$ or until $h=\ast^n$.
At each step at least a $h_i$ changes, for this reason we need at most $n$ iterations to compute the $TS_f(x)$ of a given $x$.
The process is the same if the computation is based on a perturbation $P$ because $f/P$ is considered instead of $f$.

\smallskip

To encode a configuration $x$, one can use $n$ Boolean variables where $\mathsf{x}_i$ is true iff $x_i=1$, false otherwise. 
The idea is to reproduce the iterative process explained above in which we need $n$ iterations to compute the $h=TS_f(x)$.
Recall that a subcube $h$ is an element of $\set{0,1,\ast}^n$. 
Informally, each $h_i$ can be a $0$, $1$, or both (\ie, a $\ast$).
For this reason, one can use $2n$ Boolean variables to encode a $h$. 
In particular, we will use $n$ variables $h_{(1,i)}$ and $n$ variables $h_{(0,i)}$, with $i\in\set{1,\ldots,n}$, where $h_{(1,i)}$ is true iff $h_i\in\set{1,\ast}$ and $h_{(0,i)}$ is true iff $h_i\in\set{0,\ast}$. 
At this point, if both $h_{(0,i)}$ and $h_{(1,i)}$ are true, $h_i=\ast$.

Let us denote $h^t$ the result of the $t$-th iteration.
The first clauses of our Boolean formula (which can be written in conjunctive normal form) are $h^0_{(1,i)}\iff \mathsf{x}_i$ and $h^0_{(0,i)}\iff \neg \mathsf{x}_i$. 
This corresponds to set in the beginning $TS_f(x)=x$.
To perform the first iteration, it is necessary to check the result of the update procedure over each component of $h^0$. 
Remember that we can only add $\ast$ in all iterations. 
Hence, variables that are true at one iteration will be true in all further iterations and, in the case of a $h^t_{(0,i)}$ true, giving to $h^t_{(1,i)}$ the value true corresponds to a new $\ast$ in the subcube (of course, the symmetric case also exists).
Then, $h^1_{(1,i)}$ is true if $h^0_{(1,i)}$ is true (for the reason just explained), or if according to the local function $f_i(x)=f_i(h^0)=1$. 
Following the same idea, $h^1_{(0,i)}$ is true if $h^0_{(0,i)}$ is true or if according to the local function $f_i(x)=0$.
At this point,
$$h^1_{(1,i)}\iff (h^0_{(1,i)} \lor \descbox{f_i(h^0)=1} \;) \text{, and } h^1_{(0,i)}\iff (h^0_{(0,i)} \lor \descbox{f_i(h^0)=0} \;).$$

These clauses to compute the first iteration can be generalized to compute the $(t+1)$-th iteration from the previous one. Indeed, 
$$h^{t+1}_{(1,i)}\iff (h^t_{(1,i)} \lor \descbox{\exists z\in h^t) \text{ s.t. } f_i(z)=1} \;) \text{, and }$$
$$h^{t+1}_{(0,i)}\iff (h^t_{(0,i)} \lor \descbox{\exists z\in h^t) \text{ s.t. } f_i(z)=0} \;).$$
These conditions must be translated into new clauses for our Boolean formula.
Remark that it leads to compute the $TS_f(x)$ exactly as explained before. 
However, the difference is that we will always "perform" $n$ iterations (also in the case $TS_f(x)$ does not change between two iterations at some point).
For this reason, the formula presents $2n \cdot (n+1)$ Boolean variables ($2n$ variables per iteration), and the $TS_f(x)$ can be discovered looking the values of the Boolean variables in $h^n$.

Following this idea, it is interesting to calculate the $TS_f(x)$ if the unate local functions are expressed as prepositional formulas or if the non-unate functions are expressed in a propagation complete representation.
This aspect will change just how the boxed parts of the formulas are handled. 

\subsection{Locally monotone case.} Considering a BN with unate local functions, we define $\mathsf{f}_i(h^t)$ to compute the local function $f_i$ over the $t$-th iteration.
Formally, $\mathsf{f}_i(h^t)$ is obtained from $f_i$ replacing every $\neg x_j$ with $h^t_{(0,j)}$, and every $x_j$ with $h^t_{(1,j)}$.
In this way, $\mathsf{f}_i(h^t)$ is true iff $\exists z\in h^t)$ such that $f_i(z)=1$, and false otherwise.
Similarly, $\mathsf{\overline{f}}_i(h^t)$ (obtained from $\neg f_i$ replacing every $\neg x_j$ with $h^t_{(0,j)}$, and every $x_j$ with $h^t_{(1,j)}$) is true iff $\exists z\in h^t)$ such that $f_i(z)=0$, and false otherwise. 
Replacing the boxed parts above with $\mathsf{f}_i(h^t)$ and $\mathsf{\overline{f}}_i(h^t)$, one can easily obtain the clauses of the Boolean formula to compute the smaller trap space containing $x$.
Let us point out the fact that all variables in $h^t$ depends only on variables in $h^{t-1}$ (for all $t\in\set{1,\ldots, n}$) and the variables in $h^0$ depends on the one used to encode the configuration $x$. 
The computation of a $TS_f(x)$ results in a CNF which is \emph{propagation complete} \cite{bordeaux2012knowledge} \ie, the basic unit propagation mechanism is able to deduce all the literals that are logically valid.

\subsection{General case.} 
The computation of a $TS_f(x)$ boils down to decide a finite number of times if $\exists y\in c(h^t): f_i(y)\neq h^t_i$. 
It is known that such a decision can be deterministically computed in polynomial time whenever, in general, $f_i$ are given in \emph{propagation complete} representations (BDDs, Petri nets, etc.).

A Reduced-Ordered BDD is a Directed Acyclic Graph (DAG) used to represent a Boolean formula. 
It comprises a \emph{root} node (\ie a node without incoming arcs) and a series of internal nodes, also called branch nodes, usually characterized by a name (or a value) to refer to one of the Boolean variables in the formula. 
Each node has two outgoing arcs to represent the assignment of false or true to a certain variable. 
Finally, there are two terminal nodes (also called \emph{sink nodes}) which represent the final value of the Boolean function.  
Hence, in the whole structure, a path from the root to a terminal node represents an assignment to the variables of the formula. 

According to this definition of the structure, if the non-unate local functions are given as BDDs, one is able to decide in polynomial time the result of local functions for a given configuration as it corresponds to cross the structure according to the values in the configuration. 
To evaluate the function $f_i$ over a subcube $h^t$, the idea is the same for all $j\in\set{1,\ldots,n}$ such that $h^t_{(0,j)} \lor h^t_{(1,j)}$ but, in the case of a $h_j=\ast$ (\ie, $h^t_{(0,j)} \land h^t_{(1,j)}$), we need to explore both sides of the graph. 
If at the end, it is possible to reach the sink node true, then $\exists z\in h^t)$ such that $f_i(z)=1$, and if it is possible to reach the sink node false, we know that $\exists z\in c(h^t)$ such that $f_i(z)=0$.
The exploration of the structure can be translated in propositional formulas \cite{abio2016cnf} or
in ASP \cite{chevalier2019synthesis}. 
Then, the boxed parts above can be replaced by a propositional formula that is true whether, in the BDD of the local function $i$, it is possible to obtain an update result $1$ based on the values of the variables of the $t$-th iteration, false otherwise.

\section{QBF model for the marker-reprogramming problem}\label{app:qbf}
Given the idea to encode the computation of a $TS_f(x)$ in a Boolean formula, let us present how it is possible to encode the whole reprogramming problem \eqref{p3}. 
For the sake of simplicity, we continue henceforth by considering the monotone case and its notation, but remember that it is only a choice for the explanation.

\smallskip

A Quantified Boolean Formula (QBF) in quantified conjunctive normal form (QCNF) consists of a quantifier \emph{prefix} and a CNF formula, called \emph{matrix} \cite{Bning2021TheoryOQ}.
The prefix is a sequence $Q_1V_1Q_2V_2 \ldots Q_lV_l$ of $l$ levels of quantification, where $V_1,V_2,\ldots,V_l$ are sets of pairwise distinct Boolean variables and $Q_i \in \set{\forall,\exists}$ for $i\in\set{1,\ldots,l}$. 
The reprogramming problem \eqref{p3} can be encoded in a QCNF with: 
\begin{itemize}
\item $2n$ variables to model $P$;
\item $n$ variables to model a configuration $x\in\B^n$, and $2n \cdot (n+1)$ for $TS_{f/P}(x)$;
\item $n$ variables to model $y\in TS_{f/P}(x)$, and $2n \cdot (n+1)$ for $TS_{f/P}(y)$;
\item $2n$ variables to encode the constraint $TS_{f/P}(x) \neq TS_{f/P}(y)$.
\end{itemize}
Let us present in more details the variables and the clauses involved in the QCNF.

As presented in Section \ref{sec:rep}, the problem presents $l=3$ levels of quantification.
A perturbation $P$ consists of some components of the BN clamped to Boolean values. 
Then, one can define $n$ Boolean variables to model if a component is involved in $P$ (\ie, $clamped_i$ is true with $i\in\set{1,\ldots,n}$) or not (\ie, $clamped_i$ is false). 
However, it is necessary to associate a Boolean value to the clamped components. 
Then, $n$ variables ($value_i$) can model the constant Boolean value of the clamped component $i$. 
It is clear now that, to model $P$, we use $2n$ Boolean variables existentially quantified in the first level (\ie, $Q_1=\exists$ and $V_1=\bigcup_{i\in\set{1,\ldots,n}}clamped_i \cup \bigcup_{i\in\set{1,\ldots,n}}value_i$).
As explained before, we consider the possibility to limit the number of components that can be involved in a perturbation $P$. 
For this reason, we need to add in the CNF a set of clauses to encode this constraint. 
We present, later on, how we decide to handle this aspect.

According to \eqref{p3}, in the second level of quantification, we have $Q_2=\forall$ and $V_2=\set{\mathsf{x}_1,\mathsf{x}_2,\ldots,\mathsf{x}_n}$.

At this point, we need to model the computation of the $TS_{f/P}(x)$ by editing the idea explained above for $TS_f(x)$.
The main difference, here, is that a $h^t_{(1,i)}$ is true if at least one of the following conditions hold: 
\begin{itemize}
\item $h^{t-1}_{(1,i)}$ is true (\ie, $h_i$ in the previous iteration $t-1$ is $1$ or $\ast$);
\item according to the (not perturbed) local function $f_i$ the updated value can be a $1$;
\item $clamped_i$ and $value_i$ are true (\ie, the component is forced to value $1$).
\end{itemize}  
Likewise, a $h^t_{(0,i)}$ is true if at least one of the following conditions hold: $h^{t-1}_{(0,i)}$ is true, according to the (not perturbed) local function $f_i$ the updated value can be a $0$, or $clamped_i$ is true and $value_i$ is false (\ie, the component is forced to value $0$).  
Indeed, the conditions to compute the $TS_{f/P}(x)$ are 
$$h^{t+1}_{(1,i)}\iff \Big(h^t_{(1,i)} \lor (\mathsf{f}_i(h^t) \land \neg clamped_i) \lor (clamped_i \land value_i)\Big)\text{, and}$$
$$h^{t+1}_{(0,i)}\iff \Big(h^t_{(0,i)} \lor (\mathsf{\overline{f}}_i(h^t) \land \neg clamped_i) \lor (clamped_i \land \neg value_i)\Big)$$
for all $t\in\set{1,\ldots,n}$. 
These clauses are part of the matrix, as well as $h^0_{(1,i)}\iff \mathsf{x}_i$ and $h^0_{(0,i)}\iff \neg \mathsf{x}_i$ to set in the beginning $TS_{f/P}(x)=x$.

Now, we need to consider $y$ and $TS_{f/P}(y)$. 
The configuration $y$ requires $n$ Boolean variable as $x$.
We denote these variables $\mathsf{y}_i$ with $i\in\set{1,\ldots,n}$.
However, $y$ must be a configuration in $TS_{f/P}(x))$. 
To encode this requirement, we need to add few clauses.
In particular, $\neg(\mathsf{y}_i \land \neg h^n_{(1,i)})$ and $(\mathsf{y}_i \lor h^n_{(0,i)})$. 
In fact, we cannot obtain a $y_i=1$ with $TS_{f/P}(x)_i=0$ and $y_i=0$ with $TS_{f/P}(x)_i=1$.
The computation of $TS_{f/P}(y)$ is based on the same amount of Boolean variables and clauses presented before.

Recall that in \eqref{p3}, we search a $y\in TS_{f/P}(x)$ such that $TS_{f/P}(y) \neq TS_{f/P}(x)$. 
To encode this constraint, we add in the matrix the clauses to require that $h^n_{(b,i)}\neq {h'}^n_{(b,i)}$ for at least a pair $(b,i)$ with $b\in\set{0,1}$, $i\in\set{1,\ldots,n}$, and $h'=TS_{f/P}(y)$.
The easiest way to insert this constraint is to use $2n$ variables $\text{diff}_{b,i}$ with $i\in\set{1,\ldots,n}$ and $b\in\set{0,1}$, where diff$_{b,i} \iff (h^n_{(b,i)} \neq h^{n'}_{(b,i)})$ and $\bigvee $diff$_{b,i}$.
At this point, $Q_3=\exists$ and $V_3$ contains all the remaining variables not contained in the first two quantified levels. 
 
All the variables and almost all the clauses of the matrix has been presented. 
In fact, we still need to manage the possibility to fix a maximum amount of components that can be involved in a perturbation. 
A possible approach is to initially consider the problem where no component can be perturbed.
Then, by increasing the number of components $k'$ allowed to be in $P$, we identify possible solutions to the reprogramming problem (with $k'\in\set{1,\ldots,k}$ because the given upper bound $k$ is obeyed). 
This approach is interesting for two reasons. 
First, it is well known that expressing the constraint ``at most $k$ variables'' in CNF is expensive \cite{frisch2010sat}, and second, this technique allows us to find minimal solutions. 
It is reasonable to be interested in the smallest perturbations, since they correspond to those that would require the least amount of work when we want to try to perturb the biological phenomenon, modeled in the BN, in a laboratory. 
In other words, we want to know which components, at least, need to be addressed to achieve the desired behavior.

\subsection*{Experimental evaluation}
We have implemented, in Python, a program capable of constructing the QCNF from the input (the BN is given in bnet\footnote{\url{http://colomoto.org/biolqm/doc/format-bnet.html}} format and the marker in JSON). 
The program generates the different QCNFs (gradually increasing the number of authorized components in P) in QDIMACS\footnote{\url{http://www.qbflib.org/qdimacs.html}} files which are then passed to a solver to find the solutions. 
We studied the performance using different solvers available nowadays. 
In particular, we tried CAQE \cite{rabe2015caqe}, DepQBF \cite{lonsing2017depqbf}, DQBDD \cite{dqbdd} and iDQ \cite{frohlich2014idq}. 
Some of these solvers exploit dependencies between QBF variables.
Testing the QBF approach, as well as the CEGAR one, we always considered $k\leq6$. 
To menage the uncontrollable components, the associated Boolean variables ($clamped_i$ and $value_i$) are simply not defined. 
We compared the time taken by the current approach implemented in BoNesis (based on the complementary problem) and the time taken by the QBF approach. Remark that, in this evaluation the QBF, we just considered locally monotone instances since BoNesis can be used just in this scenario.

Taking the smallest instance of the Moon (\ie, S4) dataset as an example, using the CAQE solver, it takes 1.4 seconds for the first solution and 30 seconds to enumerate them all. 
The QCNF has 802 variables (there are 2 uncontrollable components) and 5508 clauses (after enumerating the different solutions and considering the different $k$ up to 2). 
We already turn out to be, for the smallest example of the dataset, much slower than the most basic Enumeration \& Filtering approach. 
Moreover, the formula grows very fast. 
For S2 with only 4 more components, it goes to have 1322 variables and 17593 clauses.
Then, in 6 minutes, the solver cannot get the first solution. 
We have also tested solvers exploiting dependencies (such as DepQBF). 
They show to be able of improving performance in some cases, but the performance are still below Complementary.
In fact, on the 10 BNs of the dataset, the 3-QBF approach turns out to be significantly slower. 

For this reason, although it was possible to introduce a variation for the synthesis problem or implement the QBF approach for generic BNs, we first tried to improve our approach exploiting the CEGAR technique. We came to new and more efficient approach for this scenario of reprogramming over locally monotone BNs.
Nevertheless, this QBF approach has the potential to deal with non-locally monotone networks.

\clearpage
\section{Results on Trappist dataset}
\label{sec:results-trappist}

\begin{table}
\caption{Percentage of cases where we are able to find a first solution (\ie, solved) and cases where we are able to enumerate all solutions, in the reprogramming case, considering trappist instances (with $n<100$) associated with all possible markers. \label{tab:fam100}}
\centering
\begin{tabular}{|c|c|c|c|c|}
\hline
Instance & k & \# instances & Solved (UNSAT) & Enumeration  \\ \hline
\multirow[c]{4}{*}{inflammatory-bowel (n=47)} & 1 & 4 & 100\%  (0\%) & 100\% \\
 & 2 & 4 & 100\%  (0\%) & 100\% \\
 & 4 & 4 & 100\%  (0\%) & 100\% \\
 & 6 & 4 & 100\%  (0\%) & 100\% \\ \hline
\multirow[c]{4}{*}{T-LGL-survival (n=61)} & 1 & 8 & 100\%  (62\%) & 100\% \\
 & 2 & 8 & 100\%  (50\%) & 100\% \\
 & 4 & 8 & 100\%  (0\%) & 100\% \\
 & 6 & 8 & 100\%  (0\%) & 100\% \\ \hline
\multirow[c]{4}{*}{butanol-production (n=66)} & 1 & 112 & 100\%  (84\%) & 100\% \\
 & 2 & 112 & 100\%  (76\%) & 100\% \\
 & 4 & 112 & 100\%  (27\%) & 100\% \\
 & 6 & 112 & 100\%  (0\%) & 99\% \\ \hline
\multirow[c]{4}{*}{colon-cancer (n=70)} & 1 & 8 & 100\%  (38\%) & 100\% \\
 & 2 & 8 & 100\%  (12\%) & 100\% \\
 & 4 & 8 & 100\%  (12\%) & 62\% \\
 & 6 & 8 & 100\%  (12\%) & 38\% \\ \hline
\multirow[c]{4}{*}{mast-cell-activation (n=73)} & 1 & 122 & 100\%  (83\%) & 100\% \\
 & 2 & 122 & 100\%  (73\%) & 100\% \\
 & 4 & 122 & 100\%  (12\%) & 100\% \\
 & 6 & 122 & 100\%  (5\%) & 96\% \\ \hline
\multirow[c]{4}{*}{IL-6-signalling (n=86)} & 1 & 134 & 100\%  (54\%) & 100\% \\
 & 2 & 134 & 100\%  (33\%) & 100\% \\
 & 4 & 134 & 100\%  (20\%) & 100\% \\
 & 6 & 134 & 100\%  (12\%) & 96\% \\ \hline
\multirow[c]{4}{*}{Corral-ThIL-17-diff  (n=92)} & 1 & 24 & 100\%  (79\%) & 100\% \\
 & 2 & 24 & 100\%  (75\%) & 100\% \\
 & 4 & 24 & 100\%  (62\%) & 100\% \\
 & 6 & 24 & 100\%  (38\%) & 88\% \\ \hline
\multirow[c]{4}{*}{Korkut-2015 (n=99)} & 1 & 122 & 100\%  (85\%) & 100\% \\
 & 2 & 122 & 100\%  (71\%) & 100\% \\
 & 4 & 122 & 100\%  (54\%) & 67\% \\
 & 6 & 122 & 100\%  (41\%) & 51\% \\ \hline
 \end{tabular}
\end{table}

\begin{table}[p]
\caption{Percentage of cases where we are able to find a first solution (\ie, solved) and cases where we are able to enumerate all solutions, in the reprogramming case, considering trappist instances (with $100<n<1000$) associated with all possible markers. \label{tab:fam1000}}
\centering
\resizebox{\textwidth}{!}{
\begin{tabular}{|c|c|c|c|c|c|c|c|c|c|c|}
\hline
Instance & k & \# instances & Solved (UNSAT) & Enumeration & & Instance & k & \# instances & Solved (UNSAT) & Enumeration \\ \hline
\multirow[c]{4}{*}{interferon-1 (n=121)} & 1 & 42 & 100\%  (83\%) & 100\% & & \multirow[c]{4}{*}{RA-apoptosis (n=180)} & 1 & 8 & 100\%  (75\%) & 100\% \\
 & 2 & 42 & 100\%  (79\%) & 100\% & & & 2 & 8 & 100\%  (62\%) & 100\% \\
 & 4 & 42 & 100\%  (48\%) & 98\% & & & 4 & 8 & 100\%  (12\%) & 100\% \\
 & 6 & 42 & 100\%  (26\%) & 95\% & & & 6 & 8 & 100\%  (0\%) & 75\% \\ \hline
 \multirow[c]{4}{*}{adhesion-cip-migration (n=121)} & 1 & 4 & 100\%  (50\%) & 100\% & &  \multirow[c]{4}{*}{MAPK (n=181)} & 1 & 164 & 100\%  (73\%) & 100\% \\
 & 2 & 4 & 100\%  (50\%) & 100\% & & & 2 & 164 & 100\%  (55\%) & 100\% \\
 & 4 & 4 & 100\%  (0\%) & 50\% & &  & 4 & 164 & 100\%  (10\%) & 97\% \\
 & 6 & 4 & 100\%  (0\%) & 50\% & & & 6 & 164 & 100\%  (1\%) & 59\% \\ \hline
 \multirow[c]{4}{*}{TCR-TLR5-signaling (n=130)} & 1 & 124 & 100\%  (79\%) & 100\% & & \multirow[c]{4}{*}{er-stress (n=182)} & 1 & 146 & 100\%  (75\%) & 100\% \\
 & 2 & 124 & 100\%  (51\%) & 100\% & &  & 2 & 146 & 100\%  (67\%) & 100\% \\
 & 4 & 124 & 100\%  (26\%) & 94\% & &  & 4 & 146 & 100\%  (16\%) & 100\% \\
 & 6 & 124 & 100\%  (8\%) & 45\% & &  & 6 & 146 & 100\%  (5\%) & 92\% \\ \hline
 \multirow[c]{4}{*}{influenza-replication (n=131)} & 1 & 8 & 100\%  (62\%) & 100\% & & \multirow[c]{4}{*}{cascade-3 (n=183)} & 1 & 24 & 100\%  (62\%) & 100\% \\
 & 2 & 8 & 100\%  (38\%) & 100\% & & & 2 & 24 & 100\%  (42\%) & 100\% \\
 & 4 & 8 & 100\%  (38\%) & 100\% & &  & 4 & 24 & 100\%  (4\%) & 50\% \\
 & 6 & 8 & 100\%  (38\%) & 75\% & & & 6 & 24 & 100\%  (0\%) & 21\% \\ \hline
 \multirow[c]{4}{*}{prostate-cancer (n=133)} & 1 & 116 & 100\%  (91\%) & 100\% & &  \multirow[c]{4}{*}{CHO-2016 (n=200)} & 1 & 122 & 100\%  (80\%) & 100\% \\
 & 2 & 116 & 100\%  (66\%) & 100\% & &  & 2 & 122 & 100\%  (52\%) & 100\% \\
 & 4 & 116 & 100\%  (3\%) & 63\% & &  & 4 & 122 & 100\%  (2\%) & 75\% \\
 & 6 & 116 & 100\%  (0\%) & 6\% & &  & 6 & 122 & 100\%  (1\%) & 17\% \\ \hline
 \multirow[c]{4}{*}{HIV-1 (n=138)} & 1 & 150 & 100\%  (60\%) & 100\% & &  \multirow[c]{4}{*}{T-cell-check-point (n=218)} & 1 & 150 & 100\%  (63\%) & 100\% \\
 & 2 & 150 & 100\%  (49\%) & 100\% & &  & 2 & 150 & 100\%  (42\%) & 100\% \\
 & 4 & 150 & 100\%  (33\%) & 100\% & &  & 4 & 150 & 100\%  (13\%) & 79\% \\
 & 6 & 150 & 100\%  (29\%) & 96\% & &  & 6 & 150 & 100\%  (11\%) & 46\% \\ \hline
  \multirow[c]{4}{*}{fibroblasts (n=139)} & 1 & 8 & 100\%  (88\%) & 100\% & &  \multirow[c]{4}{*}{ErbB-receptor-signaling (n=247)} & 1 & 128 & 100\%  (82\%) & 100\% \\
 & 2 & 8 & 100\%  (62\%) & 100\% & &  & 2 & 128 & 100\%  (71\%) & 99\% \\
 & 4 & 8 & 100\%  (12\%) & 100\% & &  & 4 & 128 & 100\%  (36\%) & 88\% \\
 & 6 & 8 & 100\%  (0\%) & 12\% & &  & 6 & 128 & 100\%  (17\%) & 41\% \\ \hline
 \multirow[c]{4}{*}{HMOX-1-pathway (n=145)} & 1 & 152 & 100\%  (72\%) & 100\% & &  \multirow[c]{4}{*}{macrophage-activation (n=321)} & 1 & 200 & 100\%  (53\%) & 100\% \\
 & 2 & 152 & 100\%  (61\%) & 100\% & &  & 2 & 200 & 100\%  (38\%) & 100\% \\
 & 4 & 152 & 100\%  (41\%) & 100\% & &  & 4 & 200 & 100\%  (8\%) & 98\% \\
 & 6 & 152 & 100\%  (24\%) & 100\% & &  & 6 & 200 & 100\%  (4\%) & 76\% \\ \hline
 \multirow[c]{4}{*}{kynurenine-pathway (n=150)} & 1 & 184 & 100\%  (75\%) & 100\% & & \multirow[c]{4}{*}{cholocystokinin (n=383)} & 1 & 200 & 100\%  (58\%) & 100\% \\
 & 2 & 184 & 100\%  (70\%) & 100\% & &  & 2 & 200 & 100\%  (48\%) & 100\% \\
 & 4 & 184 & 100\%  (42\%) & 100\% & & & 4 & 200 & 100\%  (8\%) & 92\% \\
 & 6 & 184 & 100\%  (16\%) & 91\% & &  & 6 & 200 & 100\%  (1\%) & 44\% \\ \hline
 \multirow[c]{4}{*}{virus-replication-cycle (n=154)} & 1 & 178 & 100\%  (66\%) & 100\% & &  \multirow[c]{4}{*}{Alzheimer (n=762)} & 1 & 200 & 100\%  (58\%) & 100\% \\
 & 2 & 178 & 100\%  (52\%) & 100\% & &  & 2 & 200 & 100\%  (48\%) & 100\% \\
 & 4 & 178 & 100\%  (37\%) & 100\% & &  & 4 & 200 & 100\%  (12\%) & 94\% \\
 & 6 & 178 & 100\%  (30\%) & 96\% & &  & 6 & 200 & 78\%  (1\%) & 70\% \\ \hline
  \multirow[c]{4}{*}{immune-system (n=164)} & 1 & 138 & 100\%  (80\%) & 100\%  & & & & & & \\
 & 2 & 138 & 100\%  (58\%) & 100\% & & & & & & \\
 & 4 & 138 & 100\%  (7\%) & 82\% & & & & & & \\
 & 6 & 138 & 100\%  (0\%) & 17\% & & & & & & \\ \hline
 \end{tabular}
}
\end{table}

\begin{table}[p]
\caption{Percentage of cases where we are able to find a first solution (\ie, solved) and cases where we are able to enumerate all solutions, in the reprogramming case, considering trappist instances (with $n\geq100$) associated with all possible markers. \label{tab:fambig}}
\centering
\begin{tabular}{|c|c|c|c|c|}
\hline
Instance & k & \# instances & Solved (UNSAT) & Enumeration  \\ \hline
\multirow[c]{4}{*}{KEGG-network (n=1659)} & 1 & 200 & 100\%  (63\%) & 100\% \\
 & 2 & 200 & 100\%  (56\%) & 94\% \\
 & 4 & 200 & 100\%  (37\%) & 78\% \\
 & 6 & 200 & 100\%  (30\%) & 64\% \\ \hline
\multirow[c]{4}{*}{human-network (n=1953)} & 1 & 200 & 100\%  (64\%) & 100\% \\
 & 2 & 200 & 100\%  (54\%) & 95\% \\
 & 4 & 200 & 100\%  (30\%) & 78\% \\
 & 6 & 200 & 100\%  (27\%) & 67\% \\ \hline
\multirow[c]{4}{*}{SN-5 (n=2746)} & 1 & 200 & 99\%  (62\%) & 97\% \\
 & 2 & 200 & 95\%  (50\%) & 86\% \\
 & 4 & 200 & 94\%  (25\%) & 74\% \\
 & 6 & 200 & 94\%  (17\%) & 66\% \\ \hline
\multirow[c]{4}{*}{turei-2016 (n=4691)} & 1 & 200 & 24\%  (18\%) & 21\% \\
 & 2 & 200 & 12\%  (7\%) & 8\% \\
 & 4 & 200 & 6\%  (0\%) & 2\% \\
 & 6 & 200 & 6\%  (0\%) & 0\% \\ \hline
 \end{tabular}
\end{table}

\begin{table}[p]
\caption{Percentage of synthesis problems where we are able to find a first solution (\ie, solved) considering trappist instances with $n<100$ (associated with all possible markers).\label{tab:fam100sy}}
\centering
\begin{tabular}{|c|c|c|c|}
\hline
Instance & $|M|$ & \# instances & Solved (UNSAT)  \\ \hline
inflammatory-bowel (n=47) & 1 & 4 & 100\%  (0\%) \\ \hline
\multirow[c]{2}{*}{T-LGL-survival (n=61)} & 1 & 4 & 100\%  (0\%) \\
 & >1 & 4 & 100\%  (0\%) \\ \hline
\multirow[c]{2}{*}{butanol-production (n=66)} & 1 & 12 & 100\%  (0\%) \\
 & >1 & 100 & 100\%  (1\%) \\ \hline
\multirow[c]{2}{*}{colon-cancer (n=70)} & 1 & 4 & 100\%  (0\%) \\
 & >1 & 4 & 100\%  (0\%) \\ \hline
\multirow[c]{2}{*}{mast-cell-activation (n=73)} & 1 & 22 & 100\%  (27\%) \\
 & >1 & 100 & 99\%  (89\%) \\ \hline
\multirow[c]{2}{*}{IL-6-signalling (n=86)} & 1 & 34 & 100\%  (0\%) \\
 & >1 & 100 & 99\%  (5\%) \\ \hline
\multirow[c]{2}{*}{Corral-ThIL-17-diff  (n=92)}  & 1 & 8 & 88\%  (0\%) \\
 & >1 & 16 & 25\%  (0\%) \\ \hline
\multirow[c]{2}{*}{Korkut-2015 (n=99)} & 1 & 22 & 82\%  (0\%) \\
 & >1 & 100 & 24\%  (0\%) \\ \hline
 \end{tabular}
\end{table}

\begin{table}[p]
\caption{Percentage of synthesis problems where we are able to find a first solution (\ie, solved) considering trappist instances with $100<n<1000$ (associated with all possible markers).\label{tab:fam100sy}}
\centering
\resizebox{0.75\textwidth}{!}{
\begin{tabular}{|c|c|c|c|}
\hline
Instance & $|M|$ & \# instances & Solved (UNSAT)  \\ \hline
\multirow[c]{2}{*}{interferon-1 (n=121)} & 1 & 10 & 70\%  (40\%) \\
 & >1 & 32 & 91\%  (91\%) \\
\hline
\multirow[c]{1}{*}{adhesion-cip-migration (n=121)}  & 1 & 4 & 50\%  (0\%) \\
 \hline
\multirow[c]{2}{*}{TCR-TLR5-signaling (n=130)} & 1 & 24 & 100\%  (0\%) \\
 & >1 & 100 & 100\%  (1\%) \\ \hline
\multirow[c]{2}{*}{influenza-replication (n=131)} & 1 & 4 & 50\%  (0\%) \\
 & >1 & 4 & 0\%  (0\%) \\ \hline
\multirow[c]{2}{*}{prostate-cancer (n=133)} & 1 & 16 & 100\%  (0\%) \\
 & >1 & 100 & 76\%  (0\%) \\ \hline
\multirow[c]{2}{*}{HIV-1 (n=138)}  & 1 & 50 & 100\%  (12\%) \\
 & >1 & 100 & 95\%  (35\%) \\ \hline
 \multirow[c]{2}{*}{fibroblasts (n=139)}  & 1 & 4 & 50\%  (0\%) \\
 & >1 & 4 & 0\%  (0\%) \\ \hline
\multirow[c]{2}{*}{HMOX-1-pathway (n=145)} & 1 & 52 & 98\%  (19\%) \\
 & >1 & 100 & 85\%  (52\%) \\ \hline
\multirow[c]{2}{*}{kynurenine-pathway (n=150)}& 1 & 84 & 80\%  (45\%) \\
 & >1 & 100 & 90\%  (83\%) \\ \hline
\multirow[c]{2}{*}{virus-replication-cycle (n=154)} & 1 & 78 & 18\%  (18\%) \\
 & >1 & 100 & 31\%  (31\%) \\ \hline
 \multirow[c]{2}{*}{immune-system (n=164)} & 1 & 38 & 95\%  (11\%) \\
 & >1 & 100 & 49\%  (29\%) \\ \hline
 \multirow[c]{2}{*}{RA-apoptosis (n=180)} & 1 & 4 & 100\%  (0\%) \\
 & >1 & 4 & 75\%  (0\%) \\
 \hline
 \multirow[c]{2}{*}{MAPK (n=181)} & 1 & 64 & 78\%  (22\%) \\
 & >1 & 100 & 77\%  (52\%) \\ \hline
\multirow[c]{2}{*}{er-stress (n=182)}  & 1 & 46 & 96\%  (35\%) \\
 & >1 & 100 & 91\%  (67\%) \\ \hline
\multirow[c]{2}{*}{cascade-3 (n=183)} & 1 & 8 & 100\%  (0\%) \\
 & >1 & 16 & 88\%  (0\%) \\ \hline
 \multirow[c]{2}{*}{CHO-2016 (n=200)} & 1 & 22 & 86\%  (0\%) \\
 & >1 & 100 & 19\%  (0\%) \\ \hline
  \multirow[c]{2}{*}{T-cell-check-point (n=218)} & 1 & 50 & 100\%  (8\%) \\
 & >1 & 100 & 98\%  (28\%) \\ \hline
   \multirow[c]{2}{*}{ErbB-receptor-signaling (n=247)} & 1 & 28 & 64\%  (0\%) \\
 & >1 & 100 & 36\%  (2\%) \\ \hline
 \multirow[c]{2}{*}{macrophage-activation (n=321)} & 1 & 100 & 97\%  (97\%) \\
 & >1 & 100 & 100\%  (100\%) \\ \hline
\multirow[c]{2}{*}{cholocystokinin (n=383)} & 1 & 100 & 98\%  (28\%) \\
 & >1 & 100 & 91\%  (66\%) \\ \hline
 \end{tabular}
}
\end{table}

\begin{table}[p]
\caption{Percentage of synthesis problems where we are able to find a first solution (\ie, solved) considering trappist instances with $n\geq1000$ (associated with all possible markers).\label{tab:fambigsy}}
\centering
\begin{tabular}{|c|c|c|c|}
\hline
Instance &$|M|$ & \# instances & Solved (UNSAT)   \\ \hline
\multirow[c]{2}{*}{KEGG-network (n=1659)} & 1 & 100 & 42\%  (19\%) \\
 & >1 & 100 & 60\%  (55\%) \\ \hline
\multirow[c]{2}{*}{human-network (n=1953)} & 1 & 100 & 52\%  (38\%) \\
 & >1 & 100 & 72\%  (72\%) \\ \hline
\multirow[c]{2}{*}{SN-5 (n=2746)} & 1 & 100 & 73\%  (56\%) \\
 & >1 & 100 & 78\%  (76\%) \\ \hline
\multirow[c]{2}{*}{turei-2016 (n=4691)} & 1 & 100 & 35\%  (19\%) \\
 & >1 & 100 & 48\%  (42\%) \\ \hline
 \end{tabular}
\end{table}

\end{document}